\documentclass[11pt]{amsart}
\usepackage[foot]{amsaddr}

\usepackage[english]{babel}
\usepackage{amsmath,amsfonts,amssymb,amsthm}
\usepackage[numbers]{natbib}
\usepackage{amsthm}
\usepackage{dsfont}
\usepackage[utf8]{inputenc}
\usepackage[plain]{fancyref} %references with object type, no pages
\usepackage{braket}
\usepackage[dvipsnames]{xcolor}
\usepackage{physics}
\usepackage{amssymb}
\usepackage{algorithmicx} %Write down algorithms
\usepackage{algorithm} %Write down algorithms
\usepackage{algpseudocode}
\usepackage[left=3cm, right=3cm, top=3cm, bottom=3cm]{geometry} %More space!
\usepackage{hyperref}
\usepackage{breakurl}

\newtheorem{thm}{Theorem}[section]
\newtheorem*{thm*}{Theorem}
\newtheorem{lem}[thm]{Lemma}
\newtheorem{cor}[thm]{Corollary}
\newtheorem{ex}[thm]{Example}
\newtheorem{defi}[thm]{Definition}
\newtheorem{prop}[thm]{Proposition}

\newcommand{\A}{\mathcal{A}}
\newcommand{\BB}{\mathcal{B}}
\newcommand{\CC}{\mathcal{C}}
\newcommand{\DD}{\mathcal{D}}
\newcommand{\EE}{\mathcal{E}}
\newcommand{\FF}{\mathcal{F}}

\newcommand{\HH}{\mathcal{H}}
\newcommand{\II}{\mathcal{I}}
\newcommand{\JJ}{\mathcal{J}}
\newcommand{\KK}{\mathcal{K}}
\newcommand{\LL}{\mathcal{L}}
\newcommand{\MM}{\mathcal{M}}
\newcommand{\NN}{\mathcal{N}}
\newcommand{\OO}{\mathcal{O}}

\newcommand{\RR}{\mathcal{R}}
\renewcommand{\SS}{\mathcal{S}}
\newcommand{\TT}{\mathcal{T}}
\newcommand{\UU}{\mathcal{U}}

\newcommand{\WW}{\mathcal{W}}
\newcommand{\XX}{\mathcal{X}}
\newcommand{\YY}{\mathcal{Y}}
\newcommand{\ZZ}{\mathcal{Z}}
\newcommand{\setA}{\mathds{A}}
\newcommand{\setC}{\mathds{C}}

\newcommand{\setN}{\mathds{N}}

\newcommand{\setQ}{\mathds{Q}}
\newcommand{\setR}{\mathds{R}}

\newcommand{\cf}{cf.\ }
\newcommand{\eg}{e.g.\ }
\newcommand{\ie}{i.e.\ }

\newcommand{\idop}{\mathds{1}} %Identity Operator
\renewcommand{\tr}[1]{\mathrm{Tr}\left(#1\right)} %Trace
\renewcommand{\Tr}{\tr} %Avoids mistakingly taking the wrong \tr
\newcommand{\ptr}[2]{\mathrm{Tr}_{#1}\left(#2\right)} %partial trace
\newcommand{\id}{\mathrm{id}} %Identity map
\newcommand{\ran}{\mathrm{Ran}~} %Range
\newcommand{\diag}{\mathrm{diag}~} %Diagonal Matrix
\newcommand{\cstar}[1]{\mathrm{C}^\ast\left(#1\right)} %C^*-algebra spanned by #1
\newcommand{\bfH}{\mathbf{H}} %Homogeneous Polynomials
\newcommand{\linspan}[2][]{\mathrm{span}_{#1}\Set{#2}} %linear span
\newcommand{\opnorm}[1]{\norm{#1}_\infty} %Operator Norm
\newcommand{\herm}{\mathrm{herm}} %Subscript for Hermitean matrices
\newcommand{\sym}{\mathrm{sym}} %Subscript for symmetric matrices
\newcommand{\imI}{\mathrm{i}} % Imaginary unit
\newcommand*{\fancyrefthmlabelprefix}{thm}
\newcommand*{\fancyreflemlabelprefix}{lem}
\newcommand*{\fancyrefproplabelprefix}{prop}
\newcommand*{\fancyrefalglabelprefix}{alg}
\newcommand*{\fancyrefcorlabelprefix}{cor}
\newcommand*{\fancyrefdeflabelprefix}{def}
\frefformat{plain}{\fancyreflemlabelprefix}{lemma\fancyrefdefaultspacing#1}
\Frefformat{plain}{\fancyreflemlabelprefix}{Lemma\fancyrefdefaultspacing#1}
\frefformat{plain}{\fancyrefthmlabelprefix}{theorem\fancyrefdefaultspacing#1}
\Frefformat{plain}{\fancyrefthmlabelprefix}{Theorem\fancyrefdefaultspacing#1}
\frefformat{plain}{\fancyrefproplabelprefix}{proposition\fancyrefdefaultspacing#1}
\Frefformat{plain}{\fancyrefproplabelprefix}{Proposition\fancyrefdefaultspacing#1}
\frefformat{plain}{\fancyrefalglabelprefix}{algorithm\fancyrefdefaultspacing#1}
\Frefformat{plain}{\fancyrefalglabelprefix}{Algorithm\fancyrefdefaultspacing#1}
\frefformat{plain}{\fancyrefcorlabelprefix}{corollary\fancyrefdefaultspacing#1}
\Frefformat{plain}{\fancyrefcorlabelprefix}{Corollary\fancyrefdefaultspacing#1}
\frefformat{plain}{\fancyrefdeflabelprefix}{definition\fancyrefdefaultspacing#1}
\Frefformat{plain}{\fancyrefdeflabelprefix}{Definition\fancyrefdefaultspacing#1}
\title{Quantum compression relative to a set of measurements}
\author[A. Bluhm]{Andreas Bluhm}
\address{Technische Universität München, Zentrum Mathematik, Boltzmannstr. 3, 85748 Garching, Germany}
\email{bluhm@ma.tum.de}

\author[L. Rauber]{Lukas Rauber}
\email{lukas.rauber@tum.de}

\author[M. M. Wolf]{Michael M. Wolf}
\email{m.wolf@tum.de}

\begin{document}

%\tableofcontents

%\begin{abstract}
%In this work, we investigate the task of compressing quantum measurements in the Heisenberg picture with respect to their Hilbert space dimension. This serves to minimize data storage requirements. In our analysis, we allow for arbitrary classical side information and give a concrete procedure to determine the minimal dimension that we can compress to. Furthermore, we show that for two generic binary measurements compression is impossible, whereas for two binary von Neumann measurements we can reduce to dimension two. For the latter case we give a quantum channel with classical side information which achieves this compression. We give proofs in both an algebraic and a geometric framework.
%\end{abstract}

\begin{abstract}
In this work, we investigate the possibility of compressing a quantum system to one of smaller dimension in a way that preserves the measurement statistics of a given set of observables. In this process, we allow for an arbitrary amount of classical side information. We find that the latter can be bounded, which implies that the minimal compression dimension is stable in the sense that it cannot be decreased by allowing for small errors.
Various bounds on the minimal compression dimension are proven and an SDP-based algorithm for its computation is provided. 
The results are based on two independent approaches: an operator algebraic method using a fixed point result by Arveson and an algebro-geometric method that relies on irreducible polynomials and B{\'e}zout's theorem. The latter approach allows lifting the results from the single copy level to the case of multiple copies and from completely positive to merely positive maps.  

\end{abstract}

\maketitle

\section{Introduction}

Compression of information is essential in order to make efficient use of limited storage space or bandwidth. This is even more true if we work with quantum information for which decoherence is an existential threat that makes reliable storage or transmission an extraordinary difficult task. 

In this work, we consider the situation in which an unknown quantum state has to be stored for some time before  one out of a set of measurements is performed. We assume that this set is known beforehand and we investigate to what extent the required storage space, measured in terms of its Hilbert space dimension, can be reduced depending on the set of measurements.
 Intuitively, in this setup only the information relevant for the given set of measurements has to be preserved. So if this set is not too large and sufficiently benign, this might allow for compression that is either lossless or only introduces small errors in the measurement statistics. 
Since classical storage space is cheap compared to quantum storage, we allow for an arbitrary amount of classical side information in this process. One may envision the considered situation as part of a larger protocol, where one party has to wait for additional input that then determines the measurement to be performed. A different scenario where our analysis could be applicable is in protocols with a bounded storage assumption.

Before going into details, let us review some of the different notions of compression that appear in 
 quantum information theory and see how they relate to or differ from the setup analyzed in this paper.  A classical task is \emph{quantum source coding} \cite{Schuhmacher1995}. Here, one is given an ensemble of pure quantum states and a quantum source that prepares elements from this set with a given probability. The aim is to encode a string of these states into one of smaller length so that the original message can be retrieved up to a small error. The compression rate for which this is possible is famously bounded by the von Neumann entropy of the state describing the source and asymptotically the error can be made arbitrarily small. Hence, in this setup compression works irrespective of the measurement or operation that is eventually performed on the system. 

Another version of compression can be found in \cite{Winter2004a}. Given a quantum state and a positive operator valued measure (POVM), the task is to find another POVM acting on many copies of the state whose outcomes have fewer entropy. The new POVM is required to be close to the original one%(see the reference for a precise statement of what ``close'' means in this context)
. This amounts to reducing the number of POVM elements compared to the POVM which consists of tensor powers of the original one. The compression rate can again be bounded in terms of the entropy of the state and properties of the original POVM. These results are proven in the asymptotic setting of many copies of a given quantum state, but the results in \cite{Aubrun2016} show that similar compression is also possible in a non-asymptotic setup.
%Given any POVM, it is possible to find a sub-POVM with fewer elements which is close to the original one in distinguishability norm, where the number of elements is upper bounded in terms of the dimension of the Hilbert space in which the POVM is measured.

Instead of compressing either states or POVMs, one could also be interested in compressing both, which is the setting of \emph{model compression}. Given a set of states and  POVMs, the task is to find new states and POVMs in a Hilbert space of smaller dimension, such that the measurement statistics are unchanged, possibly up to a small error. The original and new elements need not be connected by a physical transformation. In \cite{Stark2014}, this was shown to be possible if all effect operators except one per POVM have low rank. Here, the compression is a non-linear map. 
In the same vein, lower bounds in terms of the entropy of measurement outcomes have been proven in \cite{Wehner2008}, based on random access codes.

Compared to the first two notions of compression discussed above, the setup of our paper starts with the single-copy scenario (rather than with the asymptotic case) and aims at minimizing the system size under the constraint that after decompression only the statistics of a given set of observables have to be preserved. In this respect, our setup is similar to model compression. Contrary to the setting of model compression, however, we demand both compression and decompression to be achieved by physical transformations. Moreover, we allow for an arbitrary amount of classical information, which is considered to be for free.

\section{Main results} \label{sec:results}

In this section, we will briefly outline the framework together with our main results. More detailed formulations and further results will be provided in subsequent sections. The starting point of our analysis is a set of measurements described by positive operator valued measures (POVMs). This means that one can assign a positive operator---a so-called \emph{effect operator}---to every measurable subset of outcomes. Let $\OO$ be the collection of all effect operators that belong to the considered measurements. If the underlying Hilbert space has dimension $D$, then $\OO$ is a subset of the set $\MM_D$ of complex $D\times D$ matrices. The type of compression we are interested in is given by a \emph{compression map} $\CC:\MM_D\rightarrow\MM_d\otimes\setC^n$ and a \emph{decompression map} $\DD:\MM_d\otimes\setC^n\rightarrow\MM_D$. Both maps are completely positive and trace preserving and such that for every density operator $\rho\in\MM_D$ and every effect $E\in\OO$ it holds that
\begin{equation}\label{eq:1mainresults}
\tr{(\DD\circ\CC)[\rho] E}=\tr{\rho E}.
\end{equation}
That is, we require the measurement statistics after compression and decompression to be exactly preserved. Here, $n\in\setN$ quantifies the amount of classical  information and $d\in\setN$ is the intermediate Hilbert space dimension that we want to minimize. For a given $\OO$, the minimal such dimension will be called its \emph{compression dimension}. If this equals $D$, we call $\OO$  \emph{incompressible}.

Our first finding (\Fref{lem:boundedsideinfo}) is that the amount of classical information can without loss of generality be restricted to $4 \log D$ bits. More precisely, if a map $\TT:\MM_D\rightarrow\MM_D$ can be realized as $\TT=\DD\circ\CC$ for given $n,d$, then it can be realized in this way with $n\leq D^4$ and $d$ unchanged. This fact, together with a compactness argument, then enables us to prove (\Fref{thm:inexactcase}) that the compression dimension is stable in the following sense: for every set of measurements there is an $\epsilon>0$ such that even if deviations from \Fref{eq:1mainresults} up to $\epsilon$ are allowed, the compression dimension cannot be decreased. In other words, allowing for errors does not change the picture as long as these are small enough.
In the light of this, the remaining part of the work then considers exact compression. 

We prove  bounds on the compression dimension following two different approaches: an operator algebraic and an algebro-geometric approach, to which we will for brevity refer to as algebraic and geometric, respectively. The algebraic path is based on the C$^\ast$-algebra $\cstar{\OO}$ generated by $\OO$. Being finite-dimensional, it is, up to an isomorphism, always of the form
\begin{equation*}\label{eq:summaryCstar1}
\cstar{\OO}\simeq \bigoplus_i \MM_{D_i}.
\end{equation*}
\begin{thm*}[Algebraic bounds on the compression dimension]
Let $d$ be the compression dimension of $\OO$ and $\{D_i\}$ be the dimensions of the matrix algebras occurring in the representation of the C$^\ast$-algebra $\cstar{\OO}$. Then it holds that $\min_i\{D_i\}\leq d \leq\max_i \{D_i\}$.
\end{thm*}
This is the content of \Fref{thm:newarveson} and \Fref{thm:blockcompression}. If $\OO$ for instance contains the effect operators of two binary von Neumann measurements, then these bounds generically coincide and are equal to $d=2$ if $D$ is even (\Fref{sec:twoproj}). For structureless $\OO$ with more than two elements, however, \Fref{lem:matrixalgebrameasure1} shows that the foregoing theorem implies that $d=D$, so that $\OO$ is incompressible.
 
The bounds in the foregoing theorem are tight in the sense that they cannot be improved solely on the basis of $\cstar{\OO}$ (unless the algebra consists only of one large block and one or two blocks of dimension one, in which case $d=\max_i \{D_i\}$, \cf \Fref{cor:possiblyredundant}). 
In particular, there are cases where the compression dimension is substantially smaller than the algebra generated by $\OO$.
In more abstract terms: the  C$^\ast$-algebra is too coarse and we need to resort to the operator system that is generated by $\OO$. In doing so,  the following is shown in \Fref{sec:computing}. The complexity of the corresponding algorithm is analyzed in \Fref{sec:complexity}. 
\begin{thm*}[Algorithm for the minimal compression dimension]
The compression dimension of $\OO$ is given by one of the matrix dimensions $D_i$ that occur in the representation of the C$^\ast$-algebra $\cstar{\OO}$. It can be computed by an algorithm that is based on a semidefinite program. 
\end{thm*}

The proof of correctness for this algorithm implies that the amount of classical side information needed is upper bounded by the number of matrix algebras occurring in the representation of $\cstar{\OO}$. This bound is sharper than the one on the classical side information needed for arbitrary maps of the form $\TT = \DD \circ \CC$.

The geometric approach leads to the following lower bound (\Fref{thm:geometric}):
\begin{thm*}[Algebro-geometric lower bound on the compression dimension]
Let $E_1$, $E_2 $ be in the real linear span of $\OO$ and define the real polynomial $p(x,z):=\det[x\idop - E_1 - z E_2]$. The smallest of the degrees of the irreducible factors over the reals of $p$ is a lower bound on the compression dimension of $\OO$. 
\end{thm*}
Again, if $E_1$ and $E_2$ are generic, structureless effect operators, then this lower bound is equal to $D$ (\Fref{lem:exirredpoly}). As such, the geometric lower bound turns out to be weaker than the algebraic one. However, it becomes more powerful if the setup is extended. 
For example, if several copies of the state $\rho$ are provided, the geometric argument is still valid and the lower bound remains unchanged (\cf \Fref{thm:copies}). The same is true if we allow for positive (de-)compression maps that are not necessarily completely positive (\cf \Fref{sec:generalizations}). Irrespective of the method, all our results still hold if we are only interested in preserving the expectation values of the measurements instead of the full statistics. This is true, because the elements in $\OO$ need  not be positive but only Hermitian.

Along the way, we prove some results that might be of independent interest. This includes in particular results on (Schwarz-) positive maps.

\section{Preliminaries} \label{sec:prelim}
In this section, we will review some concepts and notations from quantum information theory and classical algebraic geometry. Let $\MM_{m,n}$ for $n$, $m \in \setN$ denote the complex $m \times n$ matrices, which we concisely write as $\MM_n$ for $m = n$. The set of Hermitian $n \times n$ matrices will be written $\MM^\herm_n$; the set of real symmetric ones  $\MM^\sym_n$. For a set $\OO \subset \MM_n$, we will denote by $\cstar{\OO}$ the complex C$^\ast$-algebra generated by $\OO$ and the identity matrix $\idop$. We will also need the unitary group on $\setC^d$, $d \in \setN$, which we write as $\UU(d)$. By $\norm{\cdot}_\infty$, we denote the operator norm, whereas $\norm{A}_p$, $p \in \setN$, is the Schatten p-norm for $A \in \MM_d$. If $\Ket{\phi} \in \setC^d$, $\norm{\Ket{\phi}}_2$ is its Euclidean norm. For brevity, we will often refer to the set $\Set{1, \ldots, n}$ as $[n]$.

 We will work exclusively in  finite-dimensional settings with Hilbert space $\HH \simeq \setC^d$ for some $d \in \setN$ so that   the bounded linear operators are represented by $d \times d$ matrices with complex entries. The set of states/density operators is defined as $\SS(\setC^d) := \Set{\rho \in \MM_d : \tr{\rho} = 1, \rho \geq 0}$. Any pure state on a bipartite system $\setC^{d_A} \otimes \setC^{d_B}$, $d_A$, $d_B \in \setN$, can be expressed in terms of its Schmidt decomposition. This means that for any pure state $\Ket{\psi} \in \setC^{d_A} \otimes \setC^{d_B}$ there are orthonormal sets $\Set{\Ket{e_i}}_{i = 1}^k \subset{\setC^{d_A}}$ and $\Set{\Ket{f_j}}_{j = 1}^k \subset \setC^{d_B}$ such that 
\begin{equation*}
\Ket{\psi} = \sum_{i = 1}^k \sqrt{\lambda_i} \Ket{e_i} \otimes \Ket{f_i}
\end{equation*}
for some $\lambda_i > 0$ for all $i \in [k]$ and such that $\sum_{i = 1}^k \lambda_i = 1$. Here, $k \in \setN$ is the Schmidt rank of $\Ket{\psi}$ \cite[Proposition 2.2.1]{Keyl2002}. This concept can be extended to mixed states \cite[Definition 1]{Terhal2000}:
\begin{defi}[Schmidt number]\label{def:schmidtnumber}
A mixed state $\rho \in \SS(\setC^{d_A} \otimes \setC^{d_B})$ has Schmidt number $k$ if for any decomposition $\Set{p_i \geq 0, \Ket{\psi_i}}_{i = 1}^n$, $n \in \setN$, with
\begin{equation*}
\rho = \sum_{i = 1}^n p_i \dyad{\psi_i} 
\end{equation*}
at least one of the pure states $\Ket{\psi_i} \in \setC^{d_A} \otimes \setC^{d_B}$, $i \in [n]$, has Schmidt rank $k$ and there exists a decomposition into pure states such that every pure state has Schmidt rank at most $k$.
\end{defi}

The concept of measurement will be expressed through the set of effect operators $\EE(\setC^d): = \Set{E \in \MM^\herm_d : 0 \leq E \leq \idop}$. Let $\Sigma$ be the set of measurement outcomes, which we assume to be countable for simplicity. A set of effect operators $\Set{E_s}_{s \in \Sigma}$, $E_s \in \EE(\setC^d)$ for all $s \in \Sigma$ characterizes a positive operator valued measure (POVM) if 
\begin{equation*}
\sum_{s \in \Sigma} E_s = \idop
\end{equation*} 
(\cf \cite[Section 2.1.4]{Keyl2002}).

We describe transformations on physical systems by completely positive maps. Let $D \in \setN$. Recall that a linear map $\TT:\MM_D \to \MM_d$ is called $m$-positive if $\TT \otimes \id_m: \MM_D \otimes \MM_m \to \MM_d \otimes \MM_m$ is positive, where $\id_m$ is the identity map on $\MM_m$. $\TT$ is completely positive if it is positive for all $m \in \setN$. This is equivalent to $\TT$ having the form $\TT(A) = \sum_{i = 1}^k V_i^\ast   A V_i$, where $V_i \in \MM_{D,d}$ are the Kraus operators \cite[Section 8.2.3]{Nielsen2010}. If the map is additionally trace preserving, we will call this a quantum channel or a CPTP map. For $\TT: \MM_D \to \MM_d$ completely positive, the map $\TT^\ast  : \MM_d \to \MM_D$ will be the dual map with respect to the Hilbert-Schmidt inner product. If $\TT$ is trace preserving, the dual channel $\TT^\ast  $ is unital and furthermore it is completely positive if and only if its dual map is. We will denote by $\Ket{\Omega}$ a maximally entangled state on $\setC^{D^2}$,
\begin{equation*}
\Ket{\Omega} := \frac{1}{\sqrt{D}} \sum_{j = 1}^D \Ket{j}\otimes\Ket{j},
\end{equation*}
where $\Set{\Ket{j}}_{j = 1}^D$ is an orthonormal basis of $\setC^D$. A convenient way to check complete positivity of a linear map $\TT: \MM_D \to \MM_d$ is to compute its Choi matrix \cite{Choi1975}
\begin{equation*}
\tau = \TT \otimes \id_D(\dyad{\Omega}).
\end{equation*}
It is known that $\TT$ is completely positive if and only if $\tau$ is positive. One type of completely positive map which we will encounter frequently is the map $\Theta_{A}: \MM_D \to \MM_d$, defined as 
\begin{equation*}
\Theta_{A}(B) := A^\ast   B A \qquad \forall B \in \MM_D
\end{equation*}
for fixed $A \in \MM_{D,d}$.

Apart from completely positive maps, we will also need the notion of Schwarz maps. These are the unital positive linear maps for which the  Schwarz inequality
\begin{equation} \label{eq:schwarz}
\TT(A^\ast  ) \TT(A) \leq \TT(A^\ast   A)
\end{equation}
holds true for all $A \in \MM_D$. Note that every unital $2$-positive map (and hence also every unital completely positive map) fulfills the Schwarz inequality \cite[Proposition 3.3]{Paulsen2003}.

Furthermore, we will need some notation to work with polynomials. Let  $\setR[x_1, \ldots, x_n]$, $n \in \setN$, be the ring of polynomials in $n$-variables with real coefficients. In this work, we will only be concerned with irreducibility over the reals. Let 
\begin{equation*}
\bfH^d(n) = \Set{f \in \setR[x_1, \ldots, x_n] : f(\lambda x_1, \ldots, \lambda x_n) = \lambda^d f(x_1, \ldots, x_n) }
\end{equation*}
be the space of homogeneous polynomials in $n$ variables of degree $d$, $d \in \setN$. We will identify a polynomial with the vector of its coefficients when convenient. The set of homogeneous polynomials in $n$ variables and of any degree will be denoted by $\bfH(n) = \bigcup_{d \in \setN}\bfH^d(n)$. We recall which homogeneous polynomials are called hyperbolic: 
\begin{defi}[Hyperbolic polynomials]
Let $p \in \bfH^d(n)$. It is called hyperbolic with respect to the vector $e \in \setR^n$ if $p(e) \neq 0$ and if for all vectors $w \in \setR^n$ the univariate polynomial $t \mapsto p(w - t e) $ has only real roots.
\end{defi}
We will write $\ZZ(f)$ for the real zero set of the polynomials contained in the ideal generated by $f$. For $f \in \setR[x,y]$ ($f \in \bfH(3)$), this set will be called an algebraic curve in real (projective) space. %By $\II(U)$ we denote the ideal of all polynomials vanishing on $U$.
We can always switch between homogeneous and affine coordinates by homogenization, introducing an additional variable, and setting this additional variable to $1$, respectively (\cf \cite[$\mathsection 3$]{Bix2006}).
To conclude this section, let us finish by stating a classical result in algebraic geometry about the number of intersections of two algebraic curves (\cf \cite[Theorem 11.10]{Bix2006}).
\begin{lem}[B{\'e}zout's theorem]\label{lem:realbezout}
Let $f \in \bfH^m(3)$, $g \in \bfH^n(3)$ such that they have no common factors of positive degree over the real numbers. Then the curves $\ZZ(f)$ and $\ZZ(g)$ intersect at most $m \cdot n$ times, counting multiplicities, in the real projective plane.
\end{lem}
Then, of course, it is also true that $\ZZ(f(\cdot,1,\cdot))$ and $\ZZ(g(\cdot,1,\cdot))$ intersect in at most $m \cdot n$ points, since going to the projective plane only adds intersection points at infinity (points with $y = 0$).

\section{Setup} \label{sec:setup}

In most of this work, we will consider the following situation: We would like to perform at some later point a set of $s$ measurements, $s \in \setN$, each with countably many outcomes $\Set{a^k_i}_{i = 1}^{m_k}$, $m_k \in \setN \cup \Set{\infty}$, where the index $k$ denotes the $k$-th measurement. That is, upon preparation $\rho \in \SS(\setC^D)$ we obtain outcome $a^k_i$ with probability $\tr{\rho E^k_i}$ for all $i \in [m_k]$, $k \in [s]$. Here, $E_i^k \in \EE(\setC^D)$ is the effect operator associated to outcome $a_i^k$ and the effect operators $\Set{E_i^k}_{i = 1}^{m_k}$ belonging to the same measurement form a POVM.  Let us define the set of these effect operators
\begin{equation*}
\OO = \Set{E_i^k: i \in [m_k], k \in [s]}.
\end{equation*}
Note that assuming the outcomes to be countable simplifies notation, but our setup can easily be adapted to real measurement outcomes, for example. In that case, each effect corresponds to a measurable set of outcomes. See \cite[Section 3.1.4]{Heinosaari2012a} for details. 

 We are given an unknown quantum state $\rho \in \SS(\setC^D)$ that we want to store. In order to use a minimum of storage space, we want to keep only the information in the state relevant for the measurements that give rise to $\OO$. Motivated by the fact that classical information is cheap to store compared to quantum information,  we aim to minimize the dimension of the quantum system while allowing for an arbitrarily large amount of classical side information. Therefore, we are looking for a quantum compression channel $\CC: \MM_D \to \MM_{d}\otimes \setC^n$ and a quantum decompression channel $\DD: \MM_{d}\otimes \setC^n \to \MM_D$ such that for their composition $\TT = \DD \circ \CC$, the outcomes of the specified observables occur with the same probability as for the original state: 
\begin{equation*}
\tr{\rho E} = \tr{\TT(\rho) E}= \tr{\rho \TT^\ast  (E)} \qquad \forall \rho \in \SS(\setC^D), \forall E \in \OO.
\end{equation*}
The channels $\CC$ and $\DD$ can be seen as an instrument and a parameter dependent operation, respectively (\cf \cite[Section 3.2.5]{Keyl2002}).  %Therefore, we will consider the compression of effect operators instead of the compression of observables. 
Now we can define our notion of compression.

\begin{defi}[Compression of observables]
%Let $\Set{A_i: i \in [s]}$, $A_i \in \obs{\setC^D}$ for all $i \in [s]$ be a set of observables. These observables are
%Let $\LL(\OO)$ be an operator system. This operator system is 
Let $\OO$ be a set of Hermitian operators in $\MM_D$. The  \emph{compression dimension} of $\OO$ is the smallest $d\in\setN$ for which there is an $n\in\setN$, 
a CPTP map $\CC: \MM_D \to \MM_{d}\otimes \setC^n$ and a CPTP map $\DD: \MM_{d}\otimes \setC^n \to \MM_D$ such that for their composition $\TT = \DD \circ \CC$, the constraints 
\begin{equation}\label{eq:constraint}
\tr{\rho E} = \tr{\TT(\rho) E}= \tr{\rho \TT^\ast  (E)} \qquad \forall \rho \in \SS(\setC^D), \forall E \in \OO
\end{equation}
are satisfied. If the compression dimension equals $D$, $\OO$ is said to be \emph{incompressible}.
\end{defi}
Note that the constraints are linear, hence the relevant object is the linear subspace spanned by the effect operators, not the effect operators themselves. As the dual channel $\TT^\ast$ is unital, we can add the identity to $\OO$ without loss of generality. Then the linear subspace contains the identity operator and is therefore an operator system. Let us denote the Hermitian part of this operator system by
\begin{equation*}
\LL(\OO) := \linspan[\setR]{\OO; \idop}.
\end{equation*}
This also implies that it is irrelevant whether the effect operators belong to the same observable or to different ones, although these are two different physical situations. Therefore, we will henceforth only assume that $\OO \subset \MM_D^\herm$ instead of requiring the elements in $\OO$ to be positive or even effect operators.

\section{Approximate compression}\label{sec:inexact}

In \Fref{sec:setup}, we have demanded the measurement statistics to be exactly conserved. This may seem very restrictive, but we will see shortly that it can be relaxed without changing the picture.  The aim of this section is to show that the inexact case in which we demand 
\begin{equation*}\label{eq:inexactconstraint}
|\tr{\rho E} - \tr{\TT(\rho) E} | \leq \epsilon \qquad \forall \rho \in \SS(\setC^D), \forall E \in \OO.
\end{equation*}
instead of \Fref{eq:constraint} reduces to the exact case ($\epsilon = 0$) for $\epsilon$ small enough. This is the content of the following theorem:
\begin{thm}[Stability of compression dimension] \label{thm:inexactcase}
Let $\OO \subset \MM_D^\herm$ be a compact set with compression dimension $d$. Then there is an $\epsilon > 0$ such that for any $d^\prime < d$, $d^\prime \in \setN$ and any CPTP maps $\CC: \MM_D \to \MM_{d^\prime} \otimes \setC^n$, $\DD:\MM_{d^\prime} \otimes \setC^n \to \MM_D$ with $n \in \setN$ there is a state $\rho \in \SS(\setC^D)$ and an operator $E \in \OO$  for which 
\begin{equation*}
|\tr{\rho E} - \tr{(\DD \circ \CC)[\rho]E}| \geq \epsilon.
\end{equation*}
\end{thm}

%Using the dual channel $\TT^\ast$ in the above expression and taking the maximum over all $\rho \in \SS(\setC^D)$, we infer that we require 
%\begin{equation*}
%\opnorm{E - \TT^\ast(E)}\leq \epsilon \qquad \forall  E \in \OO.
%\end{equation*}
The compression dimension is therefore stable under small errors. To prove the statement, we will need the following lemma, which  shows that $4 \log D$ bits of classical side information suffice for compression.
\begin{lem}[Bound on classical information] \label{lem:boundedsideinfo}
Let $\CC$, $\DD$ be two CPTP maps, $\CC: \MM_D \to \MM_d \otimes \setC^n $, $\DD: \MM_d \otimes \setC^n \to \MM_D$, $n \in \setN$ and $d \leq D$. We define $\TT := \DD \circ \CC$. Then there are two CPTP maps $\widetilde{\CC}: \MM_D \to \MM_d \otimes \setC^{n_0}$, $\widetilde{\DD}: \MM_d \otimes \setC^{n_0} \to \MM_D$ with $n_0 \in \setN$, $n_0 \leq D^4$ such that $\TT = \widetilde{\DD} \circ \widetilde{\CC}$.
\end{lem}
\begin{proof}
Note that $\MM_d \otimes \setC^n \simeq \bigoplus_{i = 1}^n \MM_d$ has a block structure. Let $P_i$ be the projection onto the $i$-th block. Then $\TT_i := \DD \circ \Theta_{P_i} \circ \CC$ is again a completely positive map, although not necessarily trace preserving. The Choi matrix can thus be written
\begin{equation*}
(\TT \otimes \id)(\dyad{\Omega}) = \sum_{i = 1}^n (\TT_i \otimes \id)(\dyad{\Omega}).
\end{equation*}
We will argue that the Choi matrix has Schmidt number at most $d$ (see \Fref{def:schmidtnumber}). By the introduction of an isometry $V_i: \setC^{d} \hookrightarrow \setC^{nd}$ such that $V_i V_i^\ast = P_i$, we can decompose $\TT_i = \DD_i \circ \CC_i$ with $\CC_i: \MM_{D} \to \MM_d$ where $\CC_i = \Theta_{V_i} \circ \CC$ and $\DD_i: \MM_{d} \to \MM_D$ where $\DD_i = \DD \circ \Theta_{V_i^\ast}$. Therefore, it is easy to see that $(\CC_i \otimes \id)(\dyad{\Omega})$ has Schmidt number at most $d$. Embedding $\MM_d \hookrightarrow \MM_D$, we can regard $\DD_i$ as a map from $\MM_D$ to itself. Since it only acts on one part of the bipartite system, $\DD_i \otimes \id$ is a local operation. It is well known that such operations cannot increase the Schmidt number \cite[Proposition 1]{Terhal2000}. Hence, $(\TT_i \otimes \id)(\dyad{\Omega})$ has Schmidt number at most $d$ and the same holds for $(\TT \otimes \id)(\dyad{\Omega})$. An alternative way to see this is to note that the Kraus operators of $\CC_i$ and $\DD_i$ give a decomposition of $\TT_i$ into Kraus operators of rank at most $d$.

Now consider 
\begin{equation*}
\SS_d = \Set{\dyad{\psi} \in \SS(\setC^D \otimes \setC^D):\Ket{\psi} \mathrm{~has~Schmidt~rank~} \leq d}.
\end{equation*}
The set of states on $\setC^D \otimes \setC^D$ with Schmidt number at most $d$ can then be written as the convex hull of $\SS_d$. By Carath{\'e}odory's theorem, for every $\rho \in \SS(\setC^D \otimes \setC^D)$ of Schmidt number at most $d$ there are $D^4$ elements of $\SS_d$ such that $\rho$ can be written as a convex combination of these elements. We only need $D^4$ instead of $D^4 + 1$ elements, since $\SS(\setC^D \otimes \setC^D)$ is contained in an affine subspace of dimension $D^4-1$. That means
\begin{equation*}
(\TT \otimes \id)(\dyad{\Omega}) = \sum_{i = 1}^{D^4} p_i \dyad{\psi_i} \qquad \dyad{\psi_i} \in \SS_d,p_i \geq 0,\sum_{i = 1}^{D^4} p_i = 1.
\end{equation*}
Each $p_i \dyad{\psi_i}$ can be regarded as Choi matrix of a completely positive map $\widetilde{\TT}_i$. We would like to decompose these maps into $\widetilde{\CC}_i: \MM_D \to \MM_d \otimes \setC^{D^4}$, $\widetilde{\DD}_i: \MM_d \otimes \setC^{D^4} \to \MM_D$, $\widetilde{\TT}_i = \widetilde{\DD}_i \circ \widetilde{\CC}_i $. We note that since the Schmidt rank of $\Ket{\psi_i}$ is at most $d$, we can write it as 
\begin{equation*}
\Ket{\psi_i} = (X_i \otimes \idop) \Ket{\Omega} \qquad X_i \in \MM_{D},
\end{equation*}
where $X_i$ has rank at most $d$. We can take \eg $X_i = \sqrt{D\ptr{2}{\dyad{\psi_i}}}W$, where $W \in \UU(D)$ and $\ptr{2}{\cdot}$ denotes the partial trace over the second system. Then we can find $A_i \in \MM_{D,d}$, $B_i \in \MM_{d,D}$ such that $X_i = A_i B_i$ \cite[Theorem 0.4.6 e)]{Horn2012}. For $A_i$ we can use the polar decomposition such that $A_i = R_i Q_i$ with $Q_i \in \MM_d$, $Q_i \geq 0$ and $R_i \in \MM_{D,d}$ such that $R_i$ has orthonormal columns, which means that $R_i$ is an isometry \cite[Theorem 7.3.1 c)]{Horn2012}. Choose $\widetilde{\CC}_i: \MM_D \to \MM_d \otimes \setC^{D^4}$ as
\begin{equation*}
\widetilde{\CC}_i := p_i \Theta_{(Q_i B_i)^\ast} \otimes \dyad{i}
\end{equation*}
  and $\widetilde{\DD_i}: \MM_d \otimes \setC^{D^4} \to \MM_D$ as
\begin{equation*}
\widetilde{\DD}_i := \Theta_{R_i^\ast}\otimes \ev{\cdot}{i}.
\end{equation*}  
Then we can define $\widetilde{\CC} := \sum_{i = 1}^{D^4} \widetilde{\CC}_i$ and $\widetilde{\DD} := \sum_{i = 1}^{D^4} \widetilde{\DD}_i$, where $\Set{\Ket{i}}_{i = 1}^{D^4}$ is an orthonormal basis of $\setC^{D^4}$. The maps $\widetilde{\CC}$ and $\widetilde{\DD}$ are CPTP with $\widetilde{\DD}\circ \widetilde{\CC} = \TT$.
\end{proof}
Now we want to argue that taking the infimum over channels which arise from compression and decompression maps amounts to taking the infimum over a compact set. Define 
\begin{align*}
\mathcal{CH}_d :=& \{\TT^\ast: \MM_D \to \MM_D| \TT \mathrm{~CPTP~};  \exists \CC,\DD \mathrm{~s.t.~} \DD \circ \CC = \TT,~\CC: \MM_D \to \MM_d \otimes \setC^n,\\&~\DD: \MM_d \otimes \setC^n \to \MM_D, n \in \setN;~\CC, \DD \mathrm{~CPTP}\}\end{align*}
and 
\begin{align*}
\widetilde{\mathcal{CH}}_d :=& \{(\CC^\ast, \DD^\ast)| \CC: \MM_D \to \MM_d \otimes \setC^{D^4}, ~\DD: \MM_d \otimes \setC^{D^4} \to \MM_D;~\CC, \DD \mathrm{~CPTP}\}.
\end{align*}
\begin{lem} \label{lem:compactset}
$\widetilde{\mathcal{CH}}_d$ is a compact subset of the space $\XX := \BB(\MM_d \otimes \setC^{D^4}, \MM_D)\times \BB(\MM_D, \MM_d \otimes \setC^{D^4})$ equipped with some norm $\norm{\cdot}_{\XX}$. Here, $\BB(\HH, \KK)$ is the vector space of bounded linear operators from $\HH$ to $\KK$.
\end{lem}
\begin{proof}
Define 
\begin{equation*}
\YY_1: = \Set{\DD^\ast: \MM_D \to \MM_d \otimes \setC^{D^4}: \DD \mathrm{~CPTP}}.
\end{equation*}
This set is both closed and bounded (by the Russo-Dye theorem). Since $\BB(\MM_D,\MM_d \otimes \setC^{D^4})$ is a finite-dimensional normed space, $\YY_1$ is compact. By the same reasoning,
\begin{equation*}
\YY_2: = \Set{\CC^\ast: \MM_d \otimes \setC^{D^4} \to \MM_D: \CC \mathrm{~CPTP}}
\end{equation*} 
is compact. It can easily be seen that $\widetilde{\mathcal{CH}}_d \simeq \YY_1 \times \YY_2$. Since products of compact sets are compact again in the product topology and all our spaces are finite dimensional, the assertion follows.
\end{proof}
Now we can finally prove the main result of this section.

\begin{proof}[Proof of \Fref{thm:inexactcase}]
Consider
\begin{equation*}
\epsilon_{d^\prime} := \inf_{\TT^\ast \in \mathcal{CH}_{d^\prime}} \max_{E \in \OO} \opnorm{E - \TT^\ast(E)}.
\end{equation*}
By \Fref{lem:boundedsideinfo}, we can equivalently write 
\begin{equation} \label{eq:inf}
\epsilon_{d^\prime} := \inf_{(\CC^\ast, \DD^\ast) \in \widetilde{\mathcal{CH}}_{d^\prime}} \max_{E \in \OO} \opnorm{E - (\CC^\ast\circ \DD^\ast)(E)}.
\end{equation}
In \Fref{lem:compactset}, we have shown that $\widetilde{\mathcal{CH}}_{d^\prime}$ is a compact set. Note that $\RR \mapsto \max_{E \in \OO} \opnorm{\RR(E)}$ is a seminorm for any linear map $\RR: \MM_D \to \MM_D$ and seminorms on finite-dimensional vector spaces are continuous. It is thus easy to see that $f:(\widetilde{\mathcal{CH}}_{d^\prime}, \norm{\cdot}_{\XX}) \to \setR$, $f(\CC^\ast, \DD^\ast) = \max_{E \in \OO} \opnorm{E - (\CC^\ast\circ \DD^\ast)(E)}$ is continuous.   Therefore, the infimum in \Fref{eq:inf} is attained  and we can write
\begin{equation*}
\epsilon_{d^\prime} := \min_{(\CC^\ast, \DD^\ast) \in \widetilde{\mathcal{CH}}_{d^\prime}} \max_{E \in \OO} \opnorm{E - (\CC^\ast\circ \DD^\ast)(E)}.
\end{equation*}
Let $\epsilon := \min_{d^\prime \in [d-1]} \epsilon_{d^\prime}$.  As the compression dimension is $d$, we know that $\epsilon > 0$. This implies that for any $\TT^\ast \in \mathcal{CH}_{d^\prime}$, ${d^\prime} \in [d-1]$, there is an $E \in \OO$ such that
\begin{equation*}
\max_{\rho \in \SS(\setC^D)}|\tr{\rho E} - \tr{\TT(\rho)E}| \geq \epsilon.
\end{equation*}

\end{proof}
\section{Lower bounds} \label{sec:lower}

\subsection{Algebraic arguments} \label{sec:arveson}

In this section, we will prove and discuss a lower bound on the compression dimension using techniques from operator algebras. This lower bound will depend on the structure of the algebra which is generated by the measurements we would like to perform. Note that any finite-dimensional C$^\ast$-subalgebra of $\MM_D$ containing the identity has the form \cite[Theorem 5.6]{Farenick2001}
\begin{equation*}
U^\ast \left(\bigoplus_{i = 1}^s \MM_{D_i} \otimes \idop_{m_i}\right) U
\end{equation*}
with $\sum_{i = 1}^s D_i m_i= D$, $U \in \UU(D)$. The following theorem will be the main result of this section.

\begin{thm}[Operator algebraic lower bound on compression dimension]\label{thm:newarveson}
Let $\OO\subset\MM^\herm_D$ and \begin{equation*}
\cstar{\OO} = U^\ast \left(\bigoplus_{i = 1}^s \MM_{D_i} \otimes \idop_{m_i}\right) U,
\end{equation*}
where $\sum_{i = 1}^s D_i m_i= D$ and $U \in \UU(D)$. Then $\min_{i \in [s]}D_i$ is a lower bound on the  compression dimension of $\OO$. In particular, if $\cstar{\OO}=\MM_D$, then $\OO$ is incompressible.
\end{thm}
The proof of this statement goes back to an idea of Arveson \cite[p. 288]{Arveson1972}. In his paper, he proved the following:
\begin{lem} \label{lem:Arvesonthm}
Let $\Phi$ be a unital completely positive map of a matrix algebra $\MM_D$ onto itself whose fixed points algebraically generate the full matrix algebra. Then $\Phi$ is the identity map. 
\end{lem}
In Arveson's work, \Fref{lem:Arvesonthm} follows from a more general statement about boundary representations (\cf \cite[Theorem 2.1.1]{Arveson1972}). The proof of \Fref{thm:newarveson} uses Arveson's idea and extends it to more general situations, connecting it to the compression of quantum measurements. We start by proving a lemma which is essentially Lemma 1 on p.\ 285 f.\ in \cite{Arveson1972}. For this, we recall the definition of the support projection of a unital completely positive map. Let $\RR$ be such a map on a matrix $\ast$-algebra $\A$. Then the support projection of $\RR$ is the minimal orthogonal projection $P \in \A$ such that  $\RR(P) = \idop$.  An equivalent definition as well as basic properties of the support projection can be found in the appendix (\Fref{lem:propsuppproj}).

\begin{lem} \label{lem:commsuppproj}
Let $\RR$ be a unital completely positive linear map on a finite-dimensional C$^\ast  $-algebra $\A \subset \MM_{D}$, $D \in \setN$, such that $\RR \circ \RR = \RR$. Let $P$ be the support projection of $\RR$. Then $P$ commutes with the fixed points of $\RR$.
\end{lem}
% Norm bound necessary for Schwarz to hold
\begin{proof}
Since for positive maps $\RR(A)^\ast   = \RR(A^\ast  )$ for all $A \in \A$, this implies that the set of fixed points is closed under involution. Thus, proving $PAP = AP$ for all fixed points $A$ is enough, since it implies $AP = PA$ for self-adjoint elements and arbitrary fixed points can be decomposed into self-adjoint components. It is even sufficient to prove $PA^\ast   PAP = PA^\ast   AP$, since for any vector $\Ket{\phi}\in \setC^D$ it holds that
\begin{align*}
\norm{(\idop - P) A P \Ket{\phi}}_2^2 &= \norm{AP\Ket{\phi}}_2^2 - \norm{PAP\Ket{\phi}}_2^2 \\ &= \ev{PA^\ast  AP}{ \phi} - \ev{PA^\ast   PAP}{\phi}.
\end{align*}
By the polarization identity, this extends to all matrix elements. For the first equality, we used that $\idop - P$ is an orthogonal projection. Now let $A \in \A$ be a fixed point of $\RR$. Then $A^\ast  A \leq \RR(A^\ast  PA)$ follows from the Schwarz inequality, $\RR(A) = \RR(PA)$ and from the fact that $A$ is a fixed point. Multiplying by $P$ from both sides and using $A^\ast  PA \leq A^\ast  A$, this gives 
\begin{equation} \label{eq:arvesonineq}
PA^\ast  PAP \leq PA^\ast  AP \leq P\RR(A^\ast  PA)P
\end{equation}
This can be rewritten as $P\RR(A^\ast  PA)P - PA^\ast  PAP \geq 0$. The support projection $P$ fulfills the equation
\begin{equation*}
\RR(A) = \RR(PAP) \qquad \forall A \in \A
\end{equation*}
and $\RR|_{P\A P}$ is faithful, \ie
\begin{equation} \label{eq:faithfulsupport}
\RR(A) = 0\quad \leftrightarrow\quad PAP = 0 \qquad \forall A \in \A_+.
\end{equation}
Here, $\A_+$ are the positive elements of the algebra. This implies that 
\begin{equation*}
P\RR(A^\ast  PA)P - PA^\ast  PAP = 0.
\end{equation*}
holds since $\RR$ was assumed to be idempotent. The statement then follows from \Fref{eq:arvesonineq}.
\end{proof}

We will also need a simple proposition which allows us to consider simpler algebras. From a physicist's point of view the $\ast$-isomorphism $\pi$ takes care of the right choice of measurement basis and the elimination of duplicate blocks in the structure of the operators in $\OO$.

\begin{prop} \label{prop:simpleralgebra} Let $\OO \subset \MM_D^\herm$ be such that 
\begin{equation*}
\cstar{\OO} = U^\ast \left(\bigoplus_{i = 1}^s \MM_{D_i} \otimes \idop_{m_i}\right) U
\end{equation*}
with $\sum_{i = 1}^s D_i m_i= D$ and $U \in \UU(D)$. Then there exist unital CP maps $\pi: \MM_D \to \MM_{\sum_{i = 1}^s D_i}$ and $\pi^{-1}:\MM_{\sum_{i = 1}^s D_i} \to \MM_D$ such that
\begin{equation*}
\pi(\cstar{\OO}) = \bigoplus_{i = 1}^s \MM_{D_i}=:\A.
\end{equation*}
and $\pi|_{\cstar{\OO}}$ is a $\ast$-isomorphism with inverse $\pi^{-1}|_{\A}$. Moreover, $\OO$ can be compressed to dimension $d$ if and only if $\pi(\OO)$ can be compressed to dimension $d$.
\end{prop}
\begin{proof}
Let $A \in \cstar{\OO}$. Then it has the form
\begin{equation*}
A = U^\ast \left(\bigoplus_{i = 1}^s A_{i} \otimes \idop_{m_i}\right) U,
\end{equation*}
where $A_i \in \MM_{D_i}$. It is easy to see that $\tilde{\pi}: \cstar{\OO} \to \A$,
\begin{equation*}
\tilde{\pi}(A) = \bigoplus_{i = 1}^s A_{i},
\end{equation*}
is a $\ast$-isomorphism. Note that both $\tilde{\pi}$ and its inverse $\tilde{\pi}^{-1}$ are unital completely positive maps. Let $\EE_1: \MM_D \to \cstar{\OO}$ and $\EE_2: \MM_{\sum_{i=1}^s D_i} \to \A$ be conditional expectations onto the respective subalgebras. These maps are known to be completely positive and unital. Then $\pi = \tilde{\pi} \circ \EE_1$ and $\pi^{-1} = \tilde{\pi}^{-1} \circ \EE_2$ are the desired maps. Let $\CC^\ast: \MM_d \otimes \setC^n \to \MM_D $, $\DD^\ast: \MM_D \to \MM_d \otimes \setC^n$ be a dual compression and decompression map for $\OO$, respectively. For the constraints in \Fref{eq:constraint} to hold, $\OO$ must be in the fixed point set of $\TT^\ast = \CC^\ast \circ \DD^\ast$. Then $\pi \circ \CC^\ast$ and $\DD^\ast \circ \pi^{-1}$ are again dual channels and achieve compression to dimension $d$ for $\pi(\OO)$, because $\pi(\OO)$ is contained in the fixed point set for the composition of these maps. Conversely, let $\widetilde{\CC}^\ast: \MM_d \otimes \setC^n \to \MM_{\sum_{i = 1}^s D_i}$, $\widetilde{\DD}^\ast: \MM_{\sum_{i = 1}^s D_i} \to \MM_d \otimes \setC^n$ be a dual compression and decompression map for $\pi(\OO)$, respectively. Then by a similar argument, $\pi^{-1} \circ \widetilde{\CC}^\ast$ and $\widetilde{\DD}^\ast \circ \pi$ achieve compression to dimension $d$ for $\OO$.
\end{proof}
With these preparations, we can prove the main result of this section.
\begin{proof}[Proof of \Fref{thm:newarveson}]
%For the constraints in \Fref{eq:constraint} to hold, $\OO$ must be in the fixed point set of $\TT^\ast$. By \cite[Theorem 6.3.8]{Murphy1990}, we can assume without loss of generality that the algebra is of the form $\bigoplus_{i = 1}^s \MM_{D_i}$ for $\sum_{i = 1}^s D_i = D$, because $\A(\OO)$ must be $\ast  $-isomorphic to such an algebra. Furthermore, a $\ast  $-isomorphism between C$^\ast  $-algebras with unit is a completely positive unital map. We can therefore concatenate this $\ast  $-isomorphism $\pi: \A \mapsto \bigoplus_{i = 1}^s \MM_{D_i}$ with $\DD^\ast  $ to obtain a decompression map for the original algebra $\A$, $\widetilde{\DD}^\ast   = \DD^\ast   \circ \pi$. On the side of the states, we can also identify $\setC^D \simeq \bigoplus_{i = 1}^s  \bigoplus_{j = 1}^{m_j}\setC^{D_i}  \simeq \bigoplus_{i = 1}^s  \setC^{D_i} \otimes \setC^{m_i} $. We can then define a map $\mu: \MM_D \mapsto \MM_{\sum_{i = 1}^s D_i}$,
%\begin{equation*}
%\mu(\rho) := \ptr{\setC^{m_i}}{P_i U\rho U^\ast P_i},
%\end{equation*}
%where $P_i$ is the orthogonal projection onto $\setC^{D_i} \otimes \setC^{m_i}$. This is a CPTP map and we can concatenate $\widetilde{\CC} = \CC \circ \mu $ to obtain a compression map for the original states. It can be checked that $\tr{\rho E} = \tr{\mu(\rho)\pi(E)}$ for all $\rho \in \SS(\setC^D)$ and $E \in \OO$.

By \Fref{prop:simpleralgebra}, we can assume without loss of generality that the algebra is of the form $\bigoplus_{i = 1}^s \MM_{D_i}$ for $\sum_{i = 1}^s D_i = D$, because $\cstar{\OO}$ must be $\ast  $-isomorphic to such an algebra.  We already noted that for the constraints in \Fref{eq:constraint} to hold, $\OO$ must be in the fixed point set of $\TT^\ast = \CC^\ast \circ \DD^\ast$. %We can also assume without loss of generality that the C$^\ast$-algebra generated by the fixed points of $\TT^\ast$ is $\A(\OO)$, because otherwise we could consider $\sum_{i = 1}^s \theta_{P_i} \circ \TT^\ast$ where $P_i$ is the projection onto the $i$-th block.

Now, we note that there is an idempotent map with the same fixed points as $\TT^\ast$. We can for example consider the Ces\`aro-mean
\begin{equation*}
\TT^\ast_\infty = \lim_{N \to \infty} \frac{1}{N}\sum_{n = 1}^N (\TT^\ast)^n.
\end{equation*}
It is known that $\TT^\ast_\infty$ has the same fixed points as $\TT^\ast$, is unital, idempotent and also completely positive (\cf \Fref{lem:cesaro}). Moreover, $\TT^\ast_\infty \circ \TT^\ast = \TT^\ast_\infty$ holds.  

Now we prove that 
\begin{equation*}
\FF := \Set{A \in \MM_D:  P\TT^\ast(A)P =P AP; ~[P,A] = 0}
\end{equation*}
is a $\ast$-algebra, where $P$ is the support projection of $\TT_\infty^\ast$. We note that $\FF$ is an operator system as $\TT^\ast_\infty$ is a unital positive linear map and that $P$ commutes with $\cstar{\FF}$. Thus, we only need to show that $\FF$ is closed under multiplication. Using the Schwarz inequality and the fact that $P$ is an orthogonal projection, it follows for $A \in \FF$ that
\begin{align*}
PA^\ast  A P &= PA^\ast P  A P \\
&= P\TT^\ast(A^\ast) P   \TT^\ast(A)P \\
&\leq P\TT^\ast(A^\ast)   \TT^\ast(A)P \\
& \leq P\TT^\ast(A^\ast  A)P.
\end{align*}
Hence, we see that
\begin{equation*}
P[\TT^\ast(A^\ast  A) - A^\ast  A]P \geq 0.
\end{equation*}
Finally, we show that equality holds here. Applying $\TT^\ast_\infty$ to this and using both $\TT^\ast_\infty(PBP) = \TT^\ast_\infty(B)$ for all $B \in \MM_D$  and $\TT^\ast_\infty \circ \TT^\ast = \TT^\ast_\infty$, we infer
\begin{equation*}
\TT^\ast_\infty(P\left[\TT^\ast(A^\ast  A) - A^\ast  A\right]P) = 0.
\end{equation*}
This implies by faithfulness of $\TT^\ast_\infty|_{P \MM_D P}$ that
\begin{equation*}
P \left[\TT^\ast(A^\ast  A) - A^\ast  A\right] P = 0.
\end{equation*}
Thus, $A \in \FF$ implies $A^\ast A \in \FF$ and the fact that $\FF$ is a $\ast$-algebra then follows from the polarization identity
\begin{align*}
B^\ast A =& \frac{1}{4} [(A + B)^\ast(A + B) - (A - B)^\ast(A - B) + \imI (A + \imI B)^\ast(A + \imI B) - \imI (A - \imI B)^\ast(A - \imI B) ].
\end{align*}

The second main ingredient of the proof is the fact that the support projection $P$ of $\TT^\ast_\infty$ commutes with the fixed points of the map as shown in \Fref{lem:commsuppproj}. Then $P$ also commutes with every element of the C$^\ast  $-algebra generated by the fixed points of $\TT^\ast$. %By the considerations above, we can assume this algebra to be a direct sum of full matrix algebras.
Thus, it commutes especially with $\cstar{\OO}$. Therefore, $\cstar{\OO} \subset \FF$ and 
\begin{equation*}
P [\TT(A) - A] P = 0 \qquad \forall A \in \cstar{\OO}. 
\end{equation*}
We can now use the structure of $\cstar{\OO}$. By Schur's lemma, we can conclude that
\begin{equation*}
P = \bigoplus_{i \in [s]} \chi_{\II}(i) \idop_{\MM_{D_i}}
\end{equation*}
for some $\II \subset [s]$, where $\chi_\II$ is the indicator function of the set $\II$. Let $V_i: \setC^{D_i} \hookrightarrow \setC^D$  for $i \in [s]$ be an isometry such that $V_i V_i^\ast  $ is the projection onto the $i$-th block. As $\theta_{V_i^\ast}(B) \in \cstar{\OO}$ for all $B \in \MM_{D_i}$, we have shown that 
\begin{equation*}
(\Theta_{V_i} \circ \TT^\ast   \circ \Theta_{V_i^\ast  })(A) = A \qquad \forall A \in \MM_{D_i}, i \in \II.
\end{equation*}
Thus, $\Theta_{V_i} \circ \TT^\ast   \circ \Theta_{V_i^\ast  } = \id$ $\forall i \in \II$ holds.
By the definition of the support projection, we infer further that 
\begin{equation*}
\TT^\ast  _\infty((\idop - P) A) = \TT^\ast  _\infty(A(\idop - P)) = 0
\end{equation*}
for all $A \in \MM_D$, hence especially $0 \oplus \MM_{D_i} \oplus 0 \in \ker{\TT_\infty^\ast  }$ $\forall i \in [s] \setminus \II$. 

It could, however, be possible to enlarge the intermediate space, but to use classical side information to compress the quantum component of the system nonetheless. The following shows that this cannot happen. We identify the intermediate space $\MM_d \otimes \setC^n$ with $\bigoplus_{i = 1}^n \MM_d$. Let $Q_i$ be the orthogonal projection onto the $i$-th block, $i \in [n]$. Then 
\begin{equation*}
\TT_{ij} = \Theta_{V_j^\ast}  \circ \DD \circ \Theta_{Q_i} \circ \CC \circ \Theta_{V_j} 
\end{equation*}
is again a completely positive map and $\sum_{i = 1}^n \TT_{ij} = \id$ for every $j \in \II$. Looking at the Choi matrices for $\TT_{ij}$, we can see that each needs to be proportional to $\dyad{\Omega}$, because each Choi matrix is positive semidefinite and their sum is a rank one projection. We infer that $\TT_{ij}$ must be proportional to the identity channel, \ie $\TT_{ij} = p_i \id$, $p_i \geq 0$, $\sum_{i = 1}^n p_i = 1$. This is a well-known result in quantum information (no information without disturbance, see \cite[Section 5.2.2]{Heinosaari2012a}). From the rank-nullity theorem we conclude that $d \geq D_j$ for all $j \in \II$. As the set of fixed points of $\TT^\ast  $ is non-empty, we know that $\II$ has to be non-empty as well. From there, the lower bound on $d$ follows.
\end{proof}
The following corollary follows immediately from the proof of \Fref{thm:newarveson}.
\begin{cor}[Fixed points of Schwarz maps]\label{cor:algsub}
Let $\TT^\ast:\MM_D\rightarrow\MM_D$ be a Schwarz map and $\OO$  a set of fixed points of $\TT^\ast$ such that 
\begin{equation*}
\cstar{\OO} = \bigoplus_{i = 1}^s \MM_{D_i}
\end{equation*}
and $\sum_{i = 1}^s D_i = D$ and let $V_i: \setC^{D_i} \hookrightarrow \setC^D$ be an isometry such that $V_i V_i^\ast  $ is the projection onto the $i$-th block for $i \in [s]$.
Then there is an index set $\II \subset [s]$ such that $\Theta_{V_i} \circ \TT^\ast   \circ \Theta_{V_i^\ast  } = \id$ for all $i \in \II$ and $0 \oplus \MM_{D_i} \oplus 0 \in \ker{\TT^\ast  _\infty}$ for all $i \in [s]\setminus \II$, where $\TT^\ast  _\infty$ is the Ces\`aro-mean of $\TT^\ast  $.  Moreover, $d \geq \max_{i \in \II}D_i$.
\end{cor} 

%For generic observables, $\LL(\OO)$ will generate the full matrix algebra $\MM_D$ as a C${^\ast  }$-algebra (see \Fref{lem:matrixalgebrameasure1} for a precise statement of what we understand by generic in this case). In this case, \Fref{thm:newarveson} implies that no compression is possible. However, these effect operators are typically not the ones occurring in setups relevant for physics. These physically relevant measurements typically have more structure, \eg symmetries of the system under study. See \Fref{sec:twoproj} for a simple example of measurements which do not generate a full matrix algebra. In this case, \Fref{thm:newarveson} implies that we need to consider the block structure of the subalgebra generated by $\OO$. The dimension of the smallest block then gives a lower bound on the compression dimension.

To conclude this section, we will prove that two matrices generically generate the full matrix algebra. This shows that a set of unstructured effect operators is typically incompressible.
More precisely, we show that the set of pairs of Hermitian matrices which do not generate the full matrix algebra has measure zero.
\begin{lem} \label{lem:matrixalgebrameasure1}
Let $\NN = \Set{(A,B) \in \MM_D^\herm \times \MM_D^\herm: \cstar{\Set{A,B}} \subsetneq \MM_D}$. Then the set $\NN$ has Lebesgue measure zero on $\MM_D^\herm \times \MM_D^\herm$.
\end{lem}
\begin{proof}
By Burnside's theorem (\cf \cite{Lomonosov2004}), it is clear that $\NN$ is contained in the set of tuples of matrices which have a non-trivial common invariant subspace. This requirement can be formulated as the zero set of a polynomial as we will see. From \cite[Theorem 2.2]{George1999}, we know that if $A$, $B \in \MM^\herm_D$ have a common invariant subspace of dimension $k$, then also
\begin{equation*}
P_k(A,B) := \det[\sum_{i,j = 1}^{D-1} \comm{C_k(A)^i}{C_k(B)^j}^\ast   \comm{C_k(A)^i}{C_k(B)^j}] = 0 
\end{equation*}
where $C_k(A)$ is the $k$-th compound matrix of $A$, \ie the matrix with entries $\det(A[\alpha|\beta])$ and $\alpha$, $\beta$ sequences of strictly increasing integers contained in $[n]$, $A[\alpha|\beta]$ the submatrix of $A$ in rows $\alpha$ and columns $\beta$. The entries of $C_k(A)$ are arranged in lexicographical order. %It holds that $C_k(A^\ast  ) = C_k(A)^\ast  $ \cite[Theorem 2.1]{George1999}.
Multiplying the $P_k$, we obtain a polynomial $P := \prod_{k = 1}^{n-1} P_k$ in the real and imaginary parts of the entries of $A$, $B$ which contains $\NN$ in its zero set. Since $P$ is not identically zero, its zero set and therefore $\NN$ must have measure zero.
\end{proof}
We could also consider $\MM_D$ instead of $\MM_D^\herm$ and the statement would still hold. However, in the setting of (operator systems generated by) quantum observables, assuming the matrices involved to be Hermitian is more natural. 

% Blocks need not be in the kernel of the original, only vanish in the limit. Can give counterexamples through mixing, e.g. 1/2(\Phi(A) + B) with \Phi(E_1) = E_2

\subsection{Geometric arguments} \label{sec:geoargument}

To give a different perspective on the problem, we will prove in this section again that compression in the setup of \Fref{sec:setup} is impossible in general, this time using basic techniques from algebraic geometry. This will be useful later to obtain results in situations in which we cannot apply the techniques of \Fref{sec:arveson} (see \Fref{sec:copies}). We emphasize again that we are interested in irreducibility over the reals. The following lemma is the main technical result of this section.
\begin{lem} \label{lem:nosmalleropnorm}
Let $A$, $B \in \MM_D^\herm$ such that $p(x, z) := \det[x\idop - A - zB]$ is a polynomial of degree $D$ with a decomposition into irreducible factors
\begin{equation*}
p(x,z) = \prod_{i =1}^s p_i(x,z)^{m_i} \qquad m_i \in \setN,
\end{equation*}
where $\deg p_i = D_i$ and $\sum_{i = 1}^s m_i D_i = D$.  Moreover, let $W \subset \setR$ be open and non-empty and let $C$, $F \in \MM^\herm_d$ be such that 
\begin{equation*}
\opnorm{C + t F} = \opnorm{A + t B} \qquad \forall t \in W.
\end{equation*}
Then this implies that $d \geq \min_{i \in [s]}D_i$. 
\end{lem}
From this statement follows in particular that  $d \geq D$ if $p(x,z)$ is an irreducible polynomial. This lemma can be used to prove lower bounds on the compression dimension. 

\begin{thm}[Lower bound on compression dimension (geometric)]\label{thm:geometric}
Let $\OO \subset \MM_D^\herm$ be a set of Hermitian operators, $E_1$, $E_2 \in \LL(\OO)$ and 
\begin{equation*}
p(x,z):=\det[x \idop - E_1 - z E_2].
\end{equation*}
Then the smallest among the degrees of the irreducible factors of $p$ is a lower bound on the compression dimension of $\OO$. In particular, if $p$ is irreducible over the reals, then $\OO$ is incompressible.
\end{thm}

\begin{proof}
First, we have that $\TT^\ast$ is a contraction by the Russo-Dye theorem, since $\TT^\ast  $ is a positive unital map. The same is true for the dual channels $\DD^\ast  $, $\CC^\ast  $. If we require \Fref{eq:constraint} to hold, then $\LL(\OO)$ has to be in the fixed point space of $\TT^\ast  $ as seen before. By the fixed point property, the quantity $\opnorm{E_1 + t E_2}$ has to be preserved under $\TT^\ast  $ for all $t \in \setR$. Here, we have taken the modulus and then the maximum over all states in \Fref{eq:constraint}. Since  both $\CC^\ast$ and $\DD^\ast$ are contractions as well, this implies that 
\begin{equation*} 
\opnorm{E_1 + t E_2} = \opnorm{\DD^\ast  (E_1) + t \DD^\ast  (E_2)} \qquad \forall t \in \setR.
\end{equation*}
The assertion then follows from \Fref{lem:nosmalleropnorm}.
\end{proof}

In fact, we can strengthen  \Fref{thm:geometric} in the case when $\cstar{\OO}$ is a proper subalgebra and we have more information on its block structure. This is captured by the next corollary.

\begin{cor}
Let $\OO \subset \MM_D^\herm$ be such that
\begin{equation*}
\cstar{\OO} = U^\ast\left(\bigoplus_{i = 1}^s \MM_{D_i} \otimes \idop_{m_i}\right) U
\end{equation*}
with $\sum_{i = 1}^s D_i m_i = D$ and $U \in \UU(D)$. Then the minimal compression dimension is lower bounded by $D_{j_0}$ if there are $E_1$, $E_2 \in \OO$, $j_0 \in [s]$ and an open set $V \subset \setR$ such that 
\begin{equation*}
\opnorm{E_1 + t E_2} = \|E_1^{j_0}+ t E_2^{j_0}\|_{\infty}
\end{equation*}
for all $t \in V$ and $E_1^{j_0}$, $E_2^{j_0}$ are such that 
$\det[x \idop - E_1^{j_0} - z E_2^{j_0}]$
is irreducible over the reals.
Here, we have used that for all $E \in \OO$ we can write 
\begin{equation*}
E = U^\ast \left(\bigoplus_{j = 1}^s E^j \otimes \idop_{m_j} \right)U
\end{equation*}
for $E^j  \in \MM_{D_j}$, $j \in [s]$.
\end{cor}

\begin{proof}
As in the proof of \Fref{thm:geometric}, we obtain
\begin{equation} \label{eq:opnormconstraint}
\opnorm{E_1 + t E_2} = \opnorm{\DD^\ast  (E_1) + t \DD^\ast  (E_2)} \qquad \forall t \in \setR.
\end{equation}
%However, if $\det[x \idop + E_1 + z E_2]$ is irreducible, we know by \Fref{lem:nosmalleropnorm} that $d \geq D$. 
The definition of $\DD$ requires \begin{equation*}
\DD^\ast  (\LL(\Set{E_1, E_2})) \subset \bigoplus_{i = 1}^n \MM_{d} \simeq \MM_d \otimes \setC^n.
\end{equation*}
Assume therefore that $\DD^\ast  (E_1) + t \DD^\ast  (E_2) = \bigoplus_{i = 1}^n \left(F_1^i + t F_2^i\right)$, $F_j^i \in \MM_{d}$ for all $i \in [n]$, $j \in [2]$ and $t \in \setR$. Since 
\begin{equation*}
\opnorm{\DD^\ast  (E_1) + t \DD^\ast  (E_2)} = \max_{i \in [n]}\opnorm{F^i_1 + t F^i_2}
\end{equation*}
for a fixed $t$, we can assume that there is an open set $W \subset V$ such that 
\begin{equation} \label{eq:relevantblock}
\opnorm{\DD^\ast  (E_1) + s \DD^\ast  (E_2)} = \opnorm{F^{k_0}_1 + s F^{k_0}_2} \qquad \forall s \in W
\end{equation}
for some $k_0 \in [n]$. This is true since for two blocks either 
\begin{equation*}
\opnorm{F^{1}_1 + t F^{1}_2} = \opnorm{F^{2}_1 + t F^{2}_2}
\end{equation*}
for all $t \in V$ or there is a $t_0 \in V$ such that 
\begin{equation*}
\opnorm{F^{1}_1 + t_0 F^{1}_2} > \opnorm{F^{2}_1 + t_0 F^{2}_2}.
\end{equation*}
In the latter case, we can find an open neighborhood $W$ of $t_0$ such that
\begin{equation*}
\opnorm{F^{1}_1 + t_0 F^{1}_2} > \opnorm{F^{2}_1 + t_0 F^{2}_2}
\end{equation*}
for all $t \in W$ by continuity of the operator norm with respect to $t$. This can be extended to more blocks by induction in the block number and possibly further shrinking $W$.

By assumption, \Fref{eq:opnormconstraint} and \Fref{eq:relevantblock} then imply
\begin{equation*}
\opnorm{E_1^{j_0}+ t E_2^{j_0}} = \opnorm{F^{k_0}_1 + t F^{k_0}_2} \qquad \forall t \in W
\end{equation*}
The assertion $d \geq D_{j_0}$ then follows from \Fref{lem:nosmalleropnorm}.
\end{proof}

The condition
$\opnorm{E_1 + t E_2} = \|E_1^{j_0}+ t E_2^{j_0}\|_{\infty}$
might look artificial, but can easily be checked. We just have to find a $t_0 \in \setR$ which is not a crossing point and check which block has the largest operator norm in some open neighborhood of $t_0$. If furthermore $\mathrm{det}[x \idop - E_1^{j_0} - z E_2^{j_0}]$ is irreducible (this might be hard to check), we can apply the above corollary to find a lower bound on $d$. Note that the condition also implies that the $j_0$-th block is not redundant (\cf discussion in \Fref{sec:computing}), since \Fref{lem:nosmalleropnorm} guarantees that smaller blocks have smaller operator norm for some $t \in U$. By contractivity, it then follows that there is no unital completely positive map $\Phi: \MM_{\sum_{j = 1}^s D_j} \to \MM_{D_{j_0}}$ such that
\begin{equation*}
E_k^{j_0} = \Phi\left(\bigoplus_{i = 1}^s \chi_\II(i)E_k^{i}\right) \qquad \forall k \in [2]
\end{equation*} 
and $\II$ such that $D_i < D_{j_0}$ $\forall i \in \II$ and $\chi_\II$ is the indicator function of $\II$. We still have to prove \Fref{lem:nosmalleropnorm}, which we will do now.

\begin{proof}[Proof of \Fref{lem:nosmalleropnorm}]
First note that $A + t B$ has only real eigenvalues for $t \in \setR$. Thus, for any fixed $t$, the characteristic polynomial has $D$ real solutions counting multiplicities. Without loss of generality, let $U \subset W$ be a non-empty open set such that $\opnorm{A + t B}$ is the maximal eigenvalue $\lambda_{max}(t)$ of $A + t B$ for all $t \in U$ and the same holds for $C + t F$. We denote the maximal eigenvalue of the latter matrix by $\mu_{max}(t)$. This is possible, %since for large $t$, there are no level crossings (since for the $j$-th eigenvalue of $A + tB$, $t \lambda(B)^\downarrow_j + \lambda_n^\downarrow(A) \leq \lambda_j^\downarrow(A + tB) \leq t \lambda(B)^\downarrow_j + \lambda_1^\downarrow(A)$,  \cf \cite[Corallary III.2.2]{Bhatia1997}) 
since there are only finitely many level crossings in any finite interval (\cf \cite[p.124]{Kato1966}). Moreover, if the minimal eigenvalue of $A + t B$ has larger modulus, we can consider $-(A + t B)$ instead which clearly has the same operator norm and the same is possible for $C + t F$. Then 
\begin{equation*}
V := \Set{(x,z) : x = \lambda_{max}(z), z \in U}
\end{equation*}
is a subset of $\ZZ(p)$ with infinitely many points since $U$ is open in $\setR$. Let 
\begin{equation*}
q(x,z) := \det[x\idop - C - z D]
\end{equation*}
which is a polynomial of degree $d$. Assume $d < D_i$ for all $i \in [s]$. Since the $p_i$ are irreducible by assumption, $p_i$ and $q$ have no common factors for any $i \in [s]$. Therefore, by B\'ezout's theorem and since 
\begin{equation*}
\ZZ(p)\cap\ZZ(q) = \bigcup_{i \in [s]} (\ZZ(p_i)\cap \ZZ(q)),
\end{equation*}
the zero sets of the two polynomials have at most $\sum_{i = 1}^s d \cdot D_i$ points in common (\cf \Fref{lem:realbezout}). Thus, $\ZZ(q)$ especially cannot contain $V$, which implies $\opnorm{A + t B} \neq \opnorm{C + t F}$ for infinitely many $t \in U$, since 
\begin{equation*}
\Set{(x,z): x=\mu_{max}(z), z \in U} \subset \ZZ(q).
\end{equation*} 
\end{proof}
Let us make the following remark concerning our use of B{\'e}zout's theorem. Commonly, the theorem is formulated as an equality (counting multiplicities) over an algebraically closed field such as $\setC$. Since real polynomials are coprime over the reals if and only if they are coprime over the complex numbers (\cf \cite[Theorem 11.9]{Bix2006}), the complex version of B{\'e}zout's theorem implies an upper bound on the number of intersections of real coprime polynomials over the reals which we used here (\cf \cite[Theorem 11.10]{Bix2006}).

The last question we have to answer in this section is the existence of irreducible polynomials of any degree which arise from a determinant of $D \times D$ matrices. We would also like to know how common these are. This will also show that there are effect operators which give rise to irreducible polynomials. For this, we do not require the matrices $A$, $B \in \MM_D^\herm$ to be positive, because we can convert them into effect operators. For any $A \in \MM_D^\herm$ there is a $\lambda \in \setR$ such that $A + \lambda \idop \geq 0$ and we can scale this expression by a positive scalar such that it becomes smaller than the identity operator. This way, we can find non-zero effect operators $E_1$, $E_2$ such that $A, B \in \LL(\Set{E_1, E_2})$ and $E_1$, $E_2$ are fixed points if and only if $A$, $B$ are. Furthermore, $\det[x \idop - E_1 - z E_2]$ is irreducible if and only if $\det[x \idop - A - z B]$ is irreducible for linearly independent $A$, $B \in \LL(\Set{E_1, E_2})$, since a (non-singular) coordinate transformation does not change reducibility properties of the polynomial (\cf \cite[discussion before Theorem 4.5]{Bix2006}). The key ingredient to show existence of the required polynomials is the Lax conjecture which was proven in \cite[Conjecture 4]{Lewis2005}. We give it here for convenience.
\begin{thm}[Lax conjecture]
A polynomial $p \in \bfH^D(3)$ is hyperbolic with respect to the vector $e := (1,0,0)$ and satisfies $p(e) = 1$ if and only if there exist matrices $A,B \in \MM_D^\sym$ such that $p$ is given by
\begin{equation*}
p(x,y,z) = \det\left[ x\idop + y A + z B\right].
\end{equation*}
\end{thm}
The result that $A$, $B$ can be chosen real symmetric is even stronger than needed for our purposes.

\begin{lem} \label{lem:exirredpoly}
For any $D \in \setN$, there is an irreducible homogeneous polynomial and $A$, $B \in \MM_D^\sym$ such that 
\begin{equation*}
p(x,y,z) = \det[x \idop + y A + z B].
\end{equation*}
Moreover, these elements are generic in the set of homogeneous polynomials normalized to $p(e) = 1$ for $e := (1, 0, 0)$.
\end{lem}
\begin{proof}

By the Lax conjecture, it suffices to show that there are homogeneous polynomials of any degree which are both hyperbolic with respect to $e$ and irreducible. The case $D = 1$ is trivial, since there are no reducible elements and all polynomials are hyperbolic. %For $D = 2$, we can give a concrete example. Let $p(x,y,z) = x^2 - y^2 - z^2$. Then
%\begin{equation*}
%p(x,y,z) = \det[\idop_2 x + \begin{pmatrix} 1 & 0 \\ 0 & -1 \end{pmatrix} y +  \begin{pmatrix} 0 & 1 \\ 1 & 0 \end{pmatrix} z].
%\end{equation*}
%Any polynomial admitting a symmetric determinantal representation is hyperbolic (\eg \cite[Proposition 2]{Lewis2005}).
%Furthermore, $p(x,y,z)$ is irreducible, since it is non-singular over $\setP^2(\setC)$ and hence also irreducible over $\setC$ by B{\'e}zout's theorem \cite[Section 5.3]{Fulton1989}. 
Hence, assume $D > 1$. It is known that the set of reducible elements in this case does not contain any open subset in the Euclidean topology (see \Fref{lem:reduciblenoopensets} for a proof). Since the set of hyperbolic polynomials with respect to a fixed point $e$ has non-empty interior in this topology by \cite{Nuij1968} (\cf \Fref{sec:polyappendix} to see that this is not affected by normalization), it especially contains an open set, hence it cannot be fully contained in the set of reducible elements. Therefore, there must be elements which are both hyperbolic and irreducible. \Fref{lem:reduciblenoopensets} also states that the set of normalized reducible polynomials has measure zero, hence its intersection with the set of normalized hyperbolic polynomials has measure zero as well. %Since the latter has non-empty interior, it has non-zero Lebesgue measure, such that also the set of normalized irreducible hyperbolic polynomials has non-zero measure.
\end{proof}

\Fref{thm:geometric} states that compression is not possible if the polynomial
\begin{equation} \label{eq:detrep}
p(x,y,z) = \det[x\idop - y A - z B]
\end{equation}
is irreducible, where $A$, $B \in \LL(\Set{E_1, E_2})$. \Fref{lem:exirredpoly} therefore implies that effect operators which cannot be compressed are the generic case, \ie the set of $p(x,y,z)$ corresponding to effect operators which admit compression has Lebesgue measure zero in the space of normalized homogeneous polynomials in 3 variables of fixed degree $D$. This follows because $p$ has to be hyperbolic to admit a determinantal representation as in \Fref{eq:detrep}, even if we allow for Hermitian matrices. Furthermore, $p$ needs to be reducible to possibly admit a compression by the above. Unfortunately, irreducibility over the reals is difficult to check.

So far, we have only shown existence of such $p(x,y,z)$. We can also give an explicit example of such a polynomial in every dimension (with Hermitian matrices). 

\begin{prop} \label{prop:irredexample}
Let 
\begin{equation*}
A := \frac{1}{2}\begin{bmatrix}
0 & 1 & \ldots & 1 \\
1 & \ddots & & \vdots\\
\vdots & & \ddots & 1\\
1 & \ldots & 1 & 0
\end{bmatrix}, 
\qquad 
B := \frac{1}{2}\begin{bmatrix}
0 & \imI & \ldots & \imI \\
-\imI & \ddots & & \vdots\\
\vdots & & \ddots & \imI\\
-\imI & \ldots & -\imI & 0
\end{bmatrix},
\end{equation*}
$A$, $B \in \MM_D$, $D \geq 1$. Then the polynomial $p(x,z) := \det[x\idop + A + z B]$ is irreducible.
\end{prop}

\begin{proof}
$D = 1$ is trivial, thus assume $D \geq 2$. Reparameterizing with $\tilde{z} := z - \imI$, we obtain
\begin{equation*}
\tilde{p}(x, \tilde{z}) := \det[\tilde{A}(x) + \tilde{z} B]
\end{equation*}
with
\begin{equation*}
\tilde{A}(x) := \begin{bmatrix}
x & 0 & \ldots & 0 \\
1 & \ddots & & \vdots\\
\vdots & & \ddots & 0\\
1 & \ldots & 1 & x
\end{bmatrix}.
\end{equation*}
We know that $\tilde{p}(x,0) = x^D$. In order to prove that $\tilde{p}$ is irreducible, we show first that it cannot be decomposed as $\tilde{p} = q \cdot r$ with $q$, $r \in \setC[x,\tilde{z}]$ with $q(0,0) = 0 = r(0,0)$. Since the constant terms of both $q$ and $r$ must be zero, the expansion of $\tilde{p}(x,\epsilon)$ to first order in $\epsilon$ would have at least one root $x = 0$ if such a decomposition existed (\cf \Fref{lem:constantzero}). We can expand 
\begin{equation*}
\tilde{p}(x,\epsilon) := x^D - \frac{\epsilon \imI}{2} \left[D x^{D-1} + (x-1)^D - x^D\right] + \OO(|\epsilon|^2),
\end{equation*}
(\cf \Fref{lem:detexpansion}). However, 
\begin{equation*}
\tilde{p}(0,\epsilon) = (-1)^{D+1} \epsilon \imI/2 + \OO(|\epsilon|^2),
\end{equation*}
\ie the term linear in $\epsilon$ does not vanish. Therefore, without loss of generality, $r(0,0)$ is non-zero. This implies that $r$ is constant, since $x$ cannot divide $r(x,0)$, which implies that $q(x,0)$ must have degree $D$, and $\tilde{p}$ is of degree $D$. Thus, $\tilde{p}$ is irreducible (even over the complex numbers) and hence the same holds for $p$.
\end{proof}

\subsection{Comparing the arguments}

Before we continue, let us compare the two techniques used to prove that compression is not possible in general. We will see that the algebraic method shows incompressibility for a larger class of effect operators. However, we will see in \Fref{sec:copies} that the geometric argument can be used in situations where the algebraic argument is not applicable.

If $\cstar{\OO}$ is only a subalgebra of $\MM_D$, then $\LL(\OO) \subset U^\ast (\MM_{D_1} \oplus \MM_{D_2})U  $, with $D_1$, $D_2 \in \setN$ and $D_1 + D_2 = D$. Let $A$, $B \in \LL(\OO)$. By the above, they have the form $A = U^\ast(A_1 \oplus A_2) U$, $B = U^\ast (B_1 \oplus B_2)U$ with $A_i$, $B_i \in \MM_{D_i}$, $i \in [2]$.  Hence, 
\begin{align*}
\det[x \idop -  A - z B]& = \det[x \idop_{D_1} -  A_1 - z B_1] \det[x \idop_{D_2} -  A_2 - z B_2]\\
& = p_1(x,z) p_2(x,z) 
\end{align*}
with $p_1$, $p_2$ real polynomials of degree strictly less than $D$. Therefore, we know that $\cstar{\OO} \subsetneq \MM_D$ implies that $\det[x \idop - A - z B]$ for $A$, $B \in \LL(\OO)$ is not irreducible over the reals. We could suppose that also the converse holds, namely that for $A$, $B$ such that the above determinant is a reducible polynomial, $\cstar{\Set{A, B}}$ must be a proper subalgebra of $\MM_D$ (note that the C$^{\ast  }$-algebra does not depend on which generators were used as long as $\LL(\Set{A,B}) = \LL(\OO)$). Alas, this is not the case, as the following counterexample shows:
\begin{ex}
Let $p \in \bfH^3(3)$ be defined as 
\begin{equation*}
p(x,y,z) := (x- 1/2y)(x^2 - y^2 - z^2).
\end{equation*}
This is clearly reducible over the reals. However, $p$ admits a monic determinantal representation
\begin{equation*}
p(x,y,z) = \det[x\idop + y A + z B]
\end{equation*}
such that $\cstar{\Set{A, B}} = \MM_3$.
\end{ex}
\begin{proof}
By unitary invariance of the determinant, we can assume that $A$ is diagonal. It is easy to verify that $p(x,y,z)$ is hyperbolic with respect to $(1,0,0)$ and that $p(1,0,0) = 1$, such that we can choose $A$ to be real (\cf \cite{Lewis2005}). We can therefore compare coefficients directly and solve a system of equations for the matrix coefficients which is reasonably small. One possible determinantal representation is given by
\begin{equation*}
A = \begin{bmatrix}
-1 & & \\
& -1/2 & \\
& & ~1
\end{bmatrix} \qquad B = \begin{bmatrix}
0 & -1/2 & 0 \\
-1/2 & 0 & -\sqrt{3}/2 \\
0 & - \sqrt{3}/2 & 0
\end{bmatrix}.
\end{equation*}
The matrix $B$ has eigenvalues $-1$, $1$, $0$ with corresponding eigenvectors $(1, 2, \sqrt{3})$, $(1, -2, \sqrt{3})$ and $(-\sqrt{3}, 0, 1)$. Note that both matrices have non-degenerate spectrum. By Burnside's theorem (\cf \cite{Lomonosov2004} for the exact statement and a simple proof), the generators of any proper subalgebra of $\MM_D$ must have a common invariant subspace other than $0$ or $\setC^D$. Since the eigenvectors of $A$ and $B$ are pairwise linearly independent, there are no common invariant subspaces of dimension one. As only the eigenvector of $B$ corresponding to eigenvalue $0$ is in any of the two-dimensional subspaces spanned by the pairs of eigenvectors of $A$, there are no common two-dimensional invariant subspaces, either. By Burnside's theorem thus $\cstar{\Set{A,B}} = \MM_3$.
\end{proof}
Note that from \cite{Vinnikov1989}, we know that the determinantal representation of (irreducible smooth) algebraic curves of degree $2$ is unique up to equivalence, whereas in degree $3$, there are infinitely many (not-necessarily real symmetric) determinantal representations. Hence, it was natural to look for counterexamples of this degree.

\section{Upper bounds} \label{sec:upper}

\subsection{Compression to maximal block size}

We will show now that using classical side information we can at least compress to the dimension of the largest block. Note that the proof of the lemma yields explicit coding and decoding channels. 
\begin{thm}[Upper bound on the compression dimension] \label{thm:blockcompression}
Let $\OO \subset \MM_D^\herm$ be such that
\begin{equation} \label{eq:originalalgebra}
\cstar{\OO} = U^\ast \left(\bigoplus_{i = 1}^s \MM_{D_i} \otimes \idop_{m_i}\right) U
\end{equation}
where $\sum_{i = 1}^s D_i m_i= D$ and $U \in \UU(D)$.  Then $\max_{j \in [s]}D_j$ is an upper bound on the minimal compression dimension.
\end{thm}
\begin{proof}
By \Fref{prop:simpleralgebra}, we can assume that
\begin{equation*}
\cstar{\OO} = \bigoplus_{i = 1}^s \MM_{D_i}
\end{equation*}
with $\sum_{i = 1}^s D_i = D$. Without loss of generality, let $D_1 \geq D_j$ $\forall j \in [s]$. Let $V_j: \setC^{D_j} \hookrightarrow \setC^D$ be an isometry such that $V_jV_j^\ast   = P_j$ is the projection onto the $j$-th block. In the same vein, let $W_j: \setC^{D_j} \hookrightarrow \setC^{D_1}$ be an isometry such that $W_j W_j^\ast   = Q_j$ is the projection onto $M_{D_j}$, \ie $Q_j = \idop_{D_j} \oplus 0$. We define $\CC: \MM_D \to \MM_{D_1}\otimes \setC^s$ as
\begin{equation} \label{eq:blockcompression}
\CC(\rho) = \sum_{j = 1}^s W_j V_j^\ast   \rho V_j W_j^\ast   \otimes \dyad{j},
\end{equation}
where $\Set{\Ket{j}}_{j = 1}^s$ is an orthonormal basis of $\setC^s$. This map is obviously completely positive, since it is given in Kraus decomposition. It is also trace preserving, because
\begin{align*}
\tr{\CC(\rho)} &= \sum_{j = 1}^s \tr{W_j V_j^\ast   \rho V_j W_j^\ast   \otimes \dyad{j}} \\
&= \sum_{j = 1}^s \tr{P_j \rho} = \tr{\rho}.
\end{align*}
For $\DD$, it is easier to define the dual map. We will need the following maps $\RR_j: \MM_{D_j} \to \MM_{D_1}$ given by
\begin{equation*}
A \mapsto A \oplus \tr{A \eta_j} \idop_{D_1 - D_j} \qquad \eta_j \in \SS(\setC^{D_j}).
\end{equation*}
The choice of $\eta_j$ is somewhat arbitrary and is needed to ensure linearity. This map is completely positive, since it is a composition of $A \mapsto A \otimes \idop_2$ and the direct sum of the identity map and the map $A \mapsto \tr{A\eta_j}\idop_{D_1 - D_j}$, all of which are completely positive and unital.
%since its Choi matrix is positive semidefinite:
%\begin{align*}
%D_j \tau &=  D_j (\RR_j \otimes \id_{D_j})(\dyad{\Omega})\\
%& = D_j \dyad{\widetilde\Omega} + (0 \oplus \idop_{D_1 - D_j})\otimes \dyad{e} \geq 0.
%\end{align*}
%Here, we have defined
%\begin{equation*}
%\Ket{\widetilde{\Omega}} = \frac{1}{\sqrt{D_j}} \sum_{i = 1}^{D_j} (\Ket{i} \oplus 0_{D_1 - D_j}) \otimes \Ket{i}
%\end{equation*}
With this, we define the dual channel $\DD^\ast  : \MM_D \to \MM_{D_1}\otimes \setC^s$ as
\begin{equation} \label{eq:blockdecompression}
A \mapsto \sum_{j = 1}^s \RR_j(V_j^\ast   A V_j) \otimes \dyad{j}.
\end{equation}
This map is unital since $\RR_j$ is. To show correctness, we need to verify that $\tr{\rho E} = \tr{\CC(\rho) \DD^\ast  (E) }$ for all $\rho \in \SS(\setC^D)$, $E \in \OO$. We compute for such $\rho$, $E$
\begin{align*}
\tr{\CC(\rho) \DD^\ast  (E) } & = \sum_{j = 1}^s \tr{[V_j^\ast  \rho V_j V_j^\ast   E V_j] \oplus 0} \\
& = \sum_{j = 1}^s \tr{\rho P_j E P_j},
\end{align*}
where we used $W_j V_j^\ast \rho V_j  W_j^\ast = V_j^\ast \rho V_j \oplus 0$  in the first equation. The last line is equal to $\tr{\rho E}$ since $E$ is block diagonal.

To obtain compression and decompression maps for the original algebra $\cstar{\OO}$ in \Fref{eq:originalalgebra}, we can use the $\ast$-isomorphism given in \Fref{prop:simpleralgebra} and define $\widetilde{\DD}^\ast := \DD^\ast \circ \pi$, $\widetilde{\CC}^\ast := \pi^{-1} \circ \CC^\ast$, where $\CC$, $\DD$ are the maps constructed above.
\end{proof}
%Note that $\cstar{\idop, E_1, E_2} \subset \bigoplus_{i = 1}^s \MM_{d_i}$ is always the case if the effect operators and the identity generate only a subalgebra of a full matrix algebra, since any finite-dimensional C$^\ast  $-algebra \todo{von Neumann?} is unitarily equivalent to  
%\begin{equation*}
%0 \oplus \bigoplus_{i = 1}^k \MM_{d_k} \otimes \idop_{m_k}.
%\end{equation*}
%This is of the desired form, where we did not make use of possible duplicates, since the number of blocks does not matter, because we are not interested in the amount of classical side information. 
We have given an explicit way to compress a subalgebra to the size of its largest block. So far, it is, however, unclear if compression to the largest block is indeed the best we can do or if $d$ can be chosen smaller. Before we will pursue this, we will apply the above theorem in two concrete situations. First, we prove that for $\dim \LL(\OO) < 3$, the set of effect operators $\OO$ is trivially compressible.
\begin{prop}[Compression of a single binary measurement]
Let $\OO = \Set{E, \idop - E}$ be a set of effect operators, where $E \in \EE(\setC^D)$. Then the compression dimension is $1$.
\end{prop}
\begin{proof}
As $E$ is an effect operator, we can diagonalize $E$ to show that $\cstar{\OO}$ is $\ast  $-isomorphic to $\bigoplus_{i = 1}^s\setC$, $s \leq D$. The assertion follows from \Fref{thm:blockcompression}.
\end{proof}
We will now continue to use \Fref{thm:blockcompression} to discuss the important example of two von Neumann measurements with two outcomes each.

\subsection{Compressibility for two binary von Neumann measurements} \label{sec:twoproj}

We have shown that compressibility strongly depends on the algebra generated by the desired effect operators. In this section, we will show that in the case of two  bipartite projective measurements, we can compress to qubits ($d = 2$) using classical side information. The idea is that two projections generate an algebra which has a block structure of $2 \times 2$-matrices. This will use a finite-dimensional version of Halmos' two projections theorem (\cf \cite[Theorem 2]{Halmos1969}, \cite[Theorem 1.1]{Boettcher2010}).

Suppose we are given two orthogonal projections $P$ and $Q$ acting on a $\setC^D$ with $\ran P = M$, $\ran Q = N$. Then $\setC^D$ can be decomposed as
\begin{equation*}
\setC^D = (M \cap N) \oplus (M\cap N^\bot) \oplus (M^\bot \cap N) \oplus (M^\bot \cap N^\bot) \oplus M_0 \oplus M_1
\end{equation*}
The spaces $M_0$ and $M_1$ are defined through the decomposition of $\setC^D$ into $M$ and $M^\bot$,
\begin{align*}
M &= (M \cap N) \oplus (M\cap N^\bot) \oplus M_0 \\
M^\bot &= (M^\bot \cap N) \oplus (M^\bot \cap N^\bot) \oplus M_1,
\end{align*}
and their dimensions have to agree in order for them to be non-empty. We will use the abbreviation 
\begin{equation*}
(\alpha_1, \alpha_2, \alpha_3, \alpha_4) = \alpha_1 \idop_{M \cap N} \oplus \alpha_2 \idop_{M \cap N^\bot} \oplus \alpha_3 \idop_{M^\bot \cap N} \oplus \alpha_4 \idop_{M^\bot \cap N^\bot}.
\end{equation*}
If one of these subspaces is $\set{0}$, we will just ignore this contribution irrespective of $\alpha_j$. Note that this is the generic case. With this, we have the following theorem which is \cite[Corollary 2.2]{Boettcher2010}:

\begin{lem}\label{lem:twoproj}
If one of the spaces $M_0$ and $M_1$ is nontrivial, then these two spaces have the same dimension $r \in \setN$ and there exists a unitary matrix $V \in \MM_D$ such that 
\begin{align*}
VPV^\ast   &= (1,1,0,0) \oplus \diag \begin{bmatrix} 1 & 0 \\0 & 0 \end{bmatrix}_{j = 1}^r, \\
VQV^\ast   &= (1,0,1,0) \oplus \diag\begin{bmatrix} 1-\mu_j & \sqrt{\mu_j(1-\mu_j)}\\\sqrt{\mu_j(1-\mu_j)} & \mu_j  \end{bmatrix}_{j = 1}^r,
\end{align*}
where $0 \leq \mu_j \leq 1$ for all $j \in [r]$.
\end{lem}
This theorem is attributed to \cite[Section 2]{Wedin1983}, but similar questions concerning pairs of projections have already been studied by Camille Jordan in the 19th century. See \cite[Remark 1.3]{Boettcher2010} for a discussion of related results. Thus, the algebra generated by two projections and the identity operator consists essentially of block diagonal matrices with $2 \times 2$-blocks. For three projections, such a form can no longer be proven, since there are cases in which three projections generate the full matrix algebra (\cf concluding remarks of \cite{Boettcher2010}). Hence, we cannot guarantee compression to be possible for more than two bipartite von Neumann measurements.

\begin{prop}[Compression of two binary projective measurements]
Let $\OO = \{P, \idop - P, Q, \idop - Q\}\subset \MM_D$ be a set of effect operators and $P$, $Q$ two distinct orthogonal projections. Then the compression dimension for the set of these effect operators is upper bounded by $d = 2$.
\end{prop}

\begin{proof}
Let $P$, $Q \in \MM_D$ be two distinct orthogonal projections. \Fref{lem:twoproj} provides a unitary operator $V$ such that
\begin{align*}
Q &= V^\ast   ((1,1,0,0) \oplus Q_5 \oplus \ldots \oplus Q_k) V, \\
P &= V^\ast   ((1,0,1,0)\oplus P_5 \oplus \ldots \oplus P_k) V, 
\end{align*}
where $Q_i$, $P_j \in \MM_2$ for $i$, $j \in \Set{5, \ldots, k}$. %$k$ is less or equal the number of distinct eigenvalues of PQP \in (0, \idop)$.
We are therefore in the situation of \Fref{thm:blockcompression} with $D_i = 1$ for $i \in [4]$ and $D_j = 2$ for $j \in [k] \setminus [4]$ (identifying $\Set{(\alpha_1, \ldots, \alpha_4): \alpha_i \in \setC, i \in [4]}$ with $\setC^4$, thus eliminating redundancies). \Fref{thm:blockcompression} gives us a coding map $\CC: \MM_D \to \MM_2 \otimes \setC^k$ and a decoding map $\DD: \MM_2 \otimes \setC^k \to \MM_D$ which satisfies the constraints in \Fref{eq:constraint}.
\end{proof}

\section{Computing the compression dimension} \label{sec:computing}

Hitherto, we have only seen that the dimension of the largest block is attainable for compression (\Fref{thm:blockcompression}), whereas the dimension of the smallest block is a lower bound on the compression dimension (\Fref{thm:newarveson}), which is not necessarily attainable. In this section, we give an algorithm which allows us to compute the minimal dimension we can compress to using classical side information. We will assume that the operators in $\OO$ are already given in block diagonal form. Whether two given Hermitian operators have a common block diagonal structure can be checked using the algorithm in \cite[Section 4]{George1999}. Algorithms to bring a finite-dimensional C$^\ast  $-algebra into block diagonal form can be found \eg in \cite{Murota2010}. We analyze the latter algorithm in \Fref{sec:complexity}. Assume that we are given a set of Hermitian operators $\OO$. By \Fref{prop:simpleralgebra}, we can assume that $\cstar{\OO} = \bigoplus_{i = 1}^s \MM_{D_i}$ with $\sum_{i = 1}^s D_i = D$. The question of finding the minimal dimension which we can compress to amounts to determining which blocks are redundant, as will be proven below (\cf \Fref{thm:algorithmcorrect}). Let us define what we mean by redundant. 
\begin{defi}[Redundancy]
Let $\OO \subset \MM_D^\herm$ be such that 
\begin{equation*}
\cstar{\OO} = \bigoplus_{i = 1}^s \MM_{D_i} 
\end{equation*}
with $\sum_{i = 1}^s D_i = D$. We will call the $i$-th block \emph{redundant} if the compression dimension is smaller than $D_i$.
\end{defi}
%\begin{defi}[Redundancy]
%Let $\OO$ be a finite set of effect operators such that such that $\cstar{\OO} = \bigoplus_{i = 1}^s\MM_{d_i}$ with $\sum_{i = 1}^s d_i = D$. W.l.o.g.\ let $d_1 \geq \ldots \geq d_s$. Then we can write $E = \bigoplus_{i = 1}^s E^i$  with $E^i \in \MM_{d_i}$ for all $E \in \OO$. The $i$-th block is called \emph{redundant} (for $\OO$) if there is a unital completely positive map $\Phi: \MM_{D - \sum_{j = 1}^i} \mapsto \MM_{d_i}$ such that 
%\begin{equation*}
%\Phi\left(\begin{bmatrix}
%E^{i+1} & & \\ & \ddots & \\ & & E^s 
%\end{bmatrix}\right) = E^i \qquad \forall E \in \OO.
%\end{equation*}
%\end{defi}
We claim that checking redundancy can be phrased as an interpolation problem. Let $D_1$ be a block of maximal dimension (it does not matter which one we take if several of them have the same dimension, since all are redundant if one of them is). Then we ask whether there is a completely positive map $\Phi_1: \MM_{D} \to \MM_{D_1}$ such that 
\begin{equation*}
\Phi_1\left(\begin{bmatrix}
0& & & \\ &E^2 & & \\ & &\ddots & \\& & & E^s 
\end{bmatrix}\right) = E^1 \qquad \forall E \in \OO\cup \Set{\idop},
\end{equation*}
where $E^i \in \MM_{D_i}$ for all $i \in [s]$. This is a problem which can be solved using a semidefinite program (SDP, \cf \cite{Boyd2004}) as shown in \cite{Heinosaari2012}. Without loss of generality, we can assume that $\OO = \Set{\idop, E_2, \ldots, E_k}$, $k \in \setN$. If this is not the case, substitute $\OO$ by a set of linearly independent Hermitian operators including the identity where $k = \dim{\LL(\OO)}$. The SDP is the following:
\begin{align*}
&\mathrm{Minimize} & & \sum_{i = 1}^k \tr{(E^1_i)^T H_i}\\
&\mathrm{Subject~to}& & \sum_{i = 1}^k \begin{bmatrix}
0 & & & \\ &E_i^2 & & \\ &  &\ddots & \\ & & & E_i^s 
\end{bmatrix} \otimes H_i \geq 0 \qquad H_i \in \MM_{D_1}, \forall i \in [k]
\end{align*}
We will refer to an algorithm which solves this problem as InterpolationSDP with parameters $E_1, \ldots, E_k$ and $j$, where $j$ denotes the block which appears in the minimization (in the above case $j = 1$). This SDP has either $-\infty$ or $0$ as solution, the latter solution confirming that there is a $\Phi_1$ as specified above. If such a $\Phi_1$ cannot be found, $D_1$ is the minimal dimension we can compress to, otherwise we proceed to the next block. %we can define $\TT_1: \MM_{D - d_1} \mapsto \MM_D$ by $\GG_1 = \Phi_{1} \oplus \id_{D - d_1}$ which is completely positive and unital. 
Then we can repeat the procedure with the remaining blocks until we either encounter one block which is not redundant or we are left with only one block. This algorithm is formalized in pseudocode in \Fref{alg:mindimalgo}.

\begin{algorithm}
\caption{Compute minimal compression dimension} \label{alg:mindimalgo}
\begin{algorithmic}[1]
\Require{List of $E_i = \Set{E_i^1, \ldots, E_i^s}, E_i^j \in \MM_{D_j}, i \in [k], j \in [s]$, where $E_1 = \idop$; List $\Set{D_1, \ldots, D_s}$ s. t. $D_1 \geq \ldots \geq D_s$.}
\State $j := 1$
\State $\mathrm{Dmax} := D_1$
\While{$j < s$}
\State $h \gets \mathrm{InterpolationSDP}(E_1, \dots, E_k)(j)$ \Comment 0 if block redundant, $- \infty$ otherwise
\State $j \gets j+1$
\If{h = 0}
\State $\mathrm{Dmax} \gets D_{j}$
\State $E_i^j \gets 0$ $\forall i \in [k]$ \Comment Set largest non-zero block to zero
\Else
\State $j \gets s$ \Comment Terminates computation
\EndIf
\EndWhile
\State \Return $\mathrm{Dmax}$ \Comment Dimension of largest non-redundant block
\end{algorithmic}
\end{algorithm}

To see that there are actually Hermitian operators which give rise to redundant blocks such that the dimension we can compress to is strictly less than the maximal block dimension, we refer to the end of this section. We proceed with a proof that the dimension computed by \Fref{alg:mindimalgo} is indeed the minimal one. 
\begin{thm}[Correctness of the algorithm] \label{thm:algorithmcorrect}
The dimension computed by \Fref{alg:mindimalgo} is the compression dimension. %and the algorithm has complexity $\OO(\max\Set{k^4,D^3}s k^4 D^8 \log(kD))$.
\end{thm}
\begin{proof}
We proceed in three steps. First, we see that $d$ is the dimension of the largest block on which an optimal compression map acts as the identity. Then, we see that all larger blocks can be interpolated. Last, we see that no other blocks have a solution to the interpolation problem. Assume that we have found a map $\TT^\ast  $ such that $d$ is the compression dimension. Then this map has to be the identity on some blocks by \Fref{cor:algsub}. Let $\II \subset [s]$ be the index set of the blocks for which this is the case. Again by \Fref{cor:algsub}, we can conclude that $d \geq \max_{i \in \II}D_{i} =: D_{\max}$. We have to show that $D_{\max}$ can be attained to complete the first step, whereby $d = D_{\max}$. Since $\TT^\ast  $ and $\TT^\ast  _\infty$ have the same fixed point set, we can use $\TT^\ast  _\infty$ to construct another compression map. Note that all blocks with $i \notin \II$ lie in the kernel of $\TT^\ast  _\infty$ by \Fref{cor:algsub}. Thus, for all $j \in [s]\setminus \II$ there must be a completely positive map $\Phi_j: \MM_D \to \MM_{D_j}$ such that
\begin{equation} \label{eq:redundantblock}
\Phi_j\left(\bigoplus_{i = 1}^s \chi_{\II}(i) E^i\right) = E^j \qquad \forall E \in \OO \cup \Set{\idop},
\end{equation}
where $\chi_{\II}$ is the indicator function of the set $\II$. Hence, we can give the following compression scheme which attains $D_{\max}$. For the decompression map $\DD$, we can almost use the map given in the proof of \Fref{thm:blockcompression} with $d = D_{\max}$, but requiring the sum in \Fref{eq:blockdecompression} to run only over $\II$. Without loss of generality, we can assume that $\II$ is the set of the first $|\II|$ entries, such that $n = |\II|$ is the dimension needed for the classical side information. Let $V_i: \setC^{D_i} \hookrightarrow \setC^D$ be an isometry such that $V_i V_i^\ast   = P_i$ is the projection onto the $i$-th block $\forall i \in [s]$ and $W_j: \setC^{D_j} \hookrightarrow \setC^{D_{\max}}$ an isometry such that $W_j W_j^\ast   = Q_j$ is the projection onto $\MM_{D_j}$ $\forall j \in [s]$. Then we can define the dual compression map $\CC^\ast : \MM_{D_{\max}} \otimes \setC^n \to \MM_D$ as
\begin{equation*}
\CC^\ast  (A) := \TT^\ast  _\infty \left(\sum_{j \in \II} V_j\left(W_j^\ast   \otimes \Bra{j}\right)A\left(W_j \otimes \Ket{j}\right) V_j^\ast  \right).
\end{equation*}
This map can easily be seen to be completely positive, because $\TT^\ast  _\infty$ is. Correctness follows from the construction in \Fref{thm:blockcompression} since the missing blocks are all in the kernel of $\TT_\infty^\ast  $. The same holds for unitality. Hence, $d = D_{\max}$ since otherwise the map just defined would allow for an even better decompression, which contradicts that $d$ is minimal. This shows that all redundant blocks have a solution to the interpolation problem and completes the second step.

If we could find a set $\JJ \subset [s]$ and completely positive maps such that \Fref{eq:redundantblock} holds for this $\JJ$ instead of $\II$ and such that $D_{\max}^\prime = \max_{j \in J}D_j < d$, we could construct a dual channel attaining better compression. To see this, define a map $\RR^\ast  $ which is $\Phi_j$ on blocks $j \in [s] \setminus \JJ$ and the identity on all other blocks. Then we could substitute $\JJ$ for $\II$, $D_{\max}^\prime$ for $D_{\max}$ and $\RR^\ast  $ for $\TT^\ast  _\infty$ in the above construction to obtain a map with compression dimension $D_{\max}^\prime$, which contradicts minimality of $d$. This shows that a block admits a solution to the interpolation problem if and only if it is redundant and completes the last step.
%For the complexity, we have to perform at most $s$ steps, where each step has complexity governed by the complexity to solve the SDP. Hence $\OO(\max\Set{k^4,D^3}k^4 D^8 \log(kD))$ arithmetic operations are needed in each step (\cf \Fref{prop:complexity}). Therefore, the total complexity is as claimed.
\end{proof}
For a discussion of the complexity of the proposed algorithm, we refer to \Fref{sec:complexity}. Instead, we will show now that unless the algebra has a very specific structure, any block can  be redundant. We start with two lemmas investigating the matrix $\ast$-algebra generated by the image of a unital CP map.

\begin{lem}\label{lem:smalltobig}
Let $\A$ be a unital matrix $\ast$-algebra that contains $\MM_2$ or $\setC^3$ as a subalgebra, which we denote by $\A^\prime$. For every unital matrix $\ast$-algebra $\BB$, there is a unital CP map $\Phi: \A \to \BB$ and positive rank one elements $A_1,A_2,A_3\in\A^\prime$ such that 
\begin{enumerate}
\item $\cstar{\{A_1,A_2,A_3\}}=\A^\prime$,
\item $\cstar{\{\Phi(A_1),\Phi(A_2),\Phi(A_3)\}} = \BB$, and
\item each $\Phi(A_i)$, $i\in[3]$ is positive definite.   
\end{enumerate} 
\end{lem}
\begin{proof}
Assume that $\BB$ is not $\ast$-isomorphic to a subalgebra of $\A^\prime$ (otherwise the statement is trivial). Without loss of generality let $\BB = \bigoplus_{i=1}^s\MM_{D_i}$, as $\BB$ is $\ast$-isomorphic to such an algebra and any $\ast$-isomorphism is a unital CP map. Choose $X_i$, $Y_i \in \MM_{D_i}$ such that $0 < X_i,Y_i < \idop/2$, $\cstar{\Set{X_i, Y_i, \idop}}= \MM_{D_i}$ and such that $B_1 := \bigoplus_{i = 1}^s X_i$ and $B_2 := \bigoplus_{i = 1}^s Y_i$ both have non-degenerate spectrum. This is possible, because invoking \Fref{lem:matrixalgebrameasure1} lets us choose generic Hermitian $\tilde{X}_i$, $\tilde{Y}_i$ such that these generate the respective algebras. The non-singular transformation $X_i \mapsto \lambda_{x,i}(X_i + \mu_{x,i} \idop)$ with $\lambda_{x,i}$, $\mu_{x,i} > 0$ and $Y_i \mapsto \lambda_{y,i}(Y_i + \mu_{y,i} \idop)$ with $\lambda_{y,i}$, $\mu_{y,i} > 0$ allow us to choose the elements positive definite and not too large with an appropriate choice of parameters. 

Now set $B_3 := \idop - B_1 -B_2$. Note that $B_3 > 0$ and $\cstar{\Set{B_1, B_2, \idop}} = \BB$ hold by the above construction, which in particular implies that the $B_i$ are linearly independent if $\BB$ is non-commutative. Choose a set of linearly independent, positive rank one operators $A_i$, $i \in [3]$ such that $\sum_{j = 1}^3 A_j = \idop$ and $\tr{A_j} = c$, $c > 0$. For $\A^\prime = \setC^3$, we can pick an ONB and for $A^\prime = \MM_2$ the operators $A_i = 2/3 \dyad{a_i}$ with $\Ket{a_i} = \cos{\theta_i}\Ket{0} + \sin{\theta_i}\Ket{1}$  and $\theta_i = i 2 \pi/3$, $i \in [3]$. Note that both these choices generate $\A^\prime$ as a C$^\ast$-algebra.

We define $\tilde{\Phi}: \A^\prime \to \BB$ as
\begin{equation*}
\tilde{\Phi}(Z) = \frac{1}{c} \sum_{i = 1}^3 \tr{A_i Z}B_j.
\end{equation*}
This is clearly a unital CP map. If $\BB$ is commutative, then it follows that $D_i = 1$ for all $i \in [s]$ and it is easy to see that the assertion of the lemma holds if we extend $\tilde{\Phi}$ to a map $\Phi$ on $\A$. If $\BB$ is not commutative, we claim that $\dim(\ran(\tilde{\Phi})) = 3$, even if the preimage is restricted to the linear span of the $A_i$'s.
This can be seen as follows. Assume that $\tilde{\Phi}(A_1) = \lambda \tilde{\Phi}(A_2) + \mu \tilde{\Phi}(A_3)$ for some $\lambda$, $\mu \in \setC$. By linear independence of the $B_j$, this implies
\begin{equation*}
\tr{A_j (A_1 - \lambda A_2 - \mu A_3)} = 0 \qquad \forall j \in [3].
\end{equation*}
The above implies, however, that $\tr{|A_1 - \lambda A_2 - \mu A_3|^2} = 0$ and hence $A_1 - \lambda A_2 - \mu A_3 = 0$. This is a contradiction due to the linear independence of the $A_i$, which proves the claim that $\dim(\ran(\tilde{\Phi})) = 3$. 

Therefore, $\tilde{\Phi}$ maps $\linspan{A_1,A_2,A_3}$ onto  $\linspan{B_1, B_2, \idop}$. Let $\Phi$ be the extension of $\tilde{\Phi}$ to $\A$. So we have finally proven claim (2) of the Lemma since
\begin{equation*}
\BB\supset \cstar{\Phi(\A)} \supset \cstar{\{\Phi(A_1),\Phi(A_2),\Phi(A_3)} = \cstar{\Set{\idop, B_1, B_2}} = \BB.
\end{equation*}
\end{proof}

\begin{lem} \label{lem:toosmall}
Let $\A$ be $\ast$-isomorphic to $\setC^1$ or $\setC^2$ and let $\BB$ be non-commutative. Then there is no unital CP map $\Phi: \A \to \BB$ such that $\cstar{\Phi(\A)} = \BB$.
\end{lem}
\begin{proof}
If $\A = \setC^1$, then $\cstar{\Phi(\A)}$ is clearly commutative due to linearity of the map. The algebra is also commutative for $\A = \setC^2$ due to unitality of $\Phi$. 
\end{proof}
As a corollary to these two lemmas, we can now investigate the possible redundancies of blocks in the representation of a finite-dimensional C$^\ast$-algebra.
\begin{cor}[Tightness of algebraic bounds]\label{cor:possiblyredundant}
Let $\A = \bigoplus_{i = 1}^s \MM_{D_i}$.
\begin{enumerate}
\item If $\A $ contains three $\MM_1$-blocks in its block structure, then there is a set of effect operators $\WW \subset \A$ with compression dimension $d=1$ and s.t. $\cstar{\WW} = \A$.
\item If $\A$ contains $\MM_\delta$ for some $\delta \geq 2$ in its block structure, then we can find a set of effect operators $\WW \subset \A$ with compression dimension $d=\delta$ and such that $\cstar{\WW} = \A$.
\item Let $\delta := \max_{i \in [s]}D_i$. If $\A \setminus \MM_\delta$ does neither contain $\MM_2$ nor $\setC_3$ as subalgebra, then every $\WW$ with $\cstar{\WW} = \A$ has compression dimension $\delta$.
\end{enumerate}
\end{cor}
\begin{proof}
Claim (1) as well as claim (2) for $\delta=2$ follow 
directly from \Fref{lem:smalltobig} when choosing $\A^\prime$ as the considered subalgebra $\MM_2$ or $\setC^3$ and $\BB:=\A\setminus\A^\prime$: we define $\WW = \Set{A_i \oplus \Phi(A_i): i \in [3]}$ where $A_i$, $\Phi$ are as in the lemma. From the representation theory of matrix $\ast$-algebras it follows that $\cstar{\WW} = \A$, since we constructed the map such that $\Phi(A_i) > 0$. This excludes that the block generated by the $A_i$ has multiplicity greater than $1$ in $\cstar{\WW}$. The assertions then follows from \Fref{thm:algorithmcorrect}.
%With $W$ an isometry such that $WW^\ast$ is the projection onto the upper block, we can define $\DD^\ast = \theta_W$ and $\CC^\ast = \id \oplus \Phi$, which is the desired compression scheme. 

Claim (2) with $\delta>2$ is a simple consequence of the assertion for $\delta=2$ by using an isometric embedding of $\MM_2$ into $\MM_\delta$. The set $\WW$ is then obtained by taking the above (embedded) construction and adding sufficiently many elements of the form $A\oplus 0\in\MM_\delta\oplus (\A\setminus\MM_\delta)$ so that the $C^*$-algebra that they generate is the entire block $\MM_\delta\oplus 0$ (and not only the embedded $\MM_2$ subalgebra). We can choose one of these elements such that all blocks of dimension less than $\delta$ have operator norm strictly less than the block of dimension $\delta$. This guarantees that the compression dimension is not smaller than $\delta$ because of the contractivity of unital positive maps.

Claim (3) follows directly from \Fref{lem:toosmall} and \Fref{thm:algorithmcorrect}.
\end{proof}
Note that we can extend the above corollary to general matrix $\ast$-algebras by invoking \Fref{prop:simpleralgebra}. We have therefore shown that unless we are in the last case of the corollary, our upper and lower bounds on the compression dimension are tight.

%Note that we used the ellipsoid method to give an explicit bound for the complexity, which means that we considered the worst case. In practical instances, the algorithm is expected to have better complexity using interior point methods \cite{Boyd2004}.

\section{Generalizations} \label{sec:generalizations}

\subsection{Measurements and expectation values} \label{sec:modimeasure}

When we presented the setup in \Fref{sec:setup}, we were interested in preserving the measurement statistics, \ie the probabilities of each outcome of a fixed set of measurements. Subsequently, we realized that it makes no difference whether we assume $\OO \subset \EE(\setC^D)$ or $\OO \subset \MM_D^\herm$, since only the operator system generated by $\OO$ was important. This implies, however, that instead of (approximately) preserving the probabilities for each outcome, we could also only aim to preserve the expectation values of a set $\OO \subset \MM_D^\herm$ of observables and all results of this paper still apply. In particular, for generic $A$, $B \in \MM_D^\herm$ such that $\cstar{\set{A,B}} = \MM_D$, \Fref{thm:newarveson} still states that those observables are incompressible. 

Another possible modification of our setup would be to ask only for measurements $E^\prime$ on the compressed state $\CC(\rho)$ which return the original statistics, but without imposing that $E^\prime = \DD^\ast(E)$ for some $E \in \OO$ and some channel $\DD$. This relaxation, however, can easily be seen to allow for more powerful compression in certain cases. Let $\OO = \Set{E_1, E_2, E_3}$, where the elements form a POVM and $E_1$ and $E_2$ are generic. In the alternative setup, we can see that it is possible to compress to $d=1$ using 
\begin{equation*}
\CC(\rho) = \sum_{i=1}^3 \tr{E_i \rho} \dyad{i}
\end{equation*}
and choosing $E^\prime_i = \dyad{i}$. However, \Fref{thm:newarveson} implies that $\OO$ is incompressible in the original setup. The explanation for this difference is that there is no channel $\DD$ which allows to map the elements of $\OO$ to projections. Therefore, we see that this modification changes the problem significantly and we leave it for future work. 

\subsection{Positive and Schwarz maps}

The aim of this section is to explore how much we can relax the requirements on the compression and decompression channel. We still consider the setup of \Fref{sec:setup}, but now we require $\CC$, $\DD$ only to be positive instead of completely positive. Since the argument at the beginning of the proof of \Fref{thm:geometric} uses only that the maps involved are positive and trace preserving to apply the Russo-Dye theorem, the results of \Fref{thm:geometric} carry over to this setting. Note that the results obtained in the algebraic setting do not carry over to arbitrary positive maps, since for Arveson's result it is important that the map is a Schwarz map (see remark before \cite[Example 5.3]{Wolf2012}). Complete positivity, however, is not needed; a trace preserving positive map whose dual is also a Schwarz map is enough. See \Fref{lem:schwarzcesaro} for a proof that the Ces{\`a}ro-mean of a Schwarz map is again a Schwarz map. Using  \Fref{lem:noinfo} instead of the corresponding well-known result for completely positive maps, we can extend \Fref{thm:newarveson} and \Fref{lem:Arvesonthm} to Schwarz maps. From a physicist's perspective exchanging $\DD$ for a positive instead of a completely positive map can be interpreted as measuring different effect operators and inferring from them the statistics with respect to the original effect operators. Note that this is still less general than the modified setup discussed in \Fref{sec:modimeasure}. Since only completely positive maps are considered meaningful evolutions of a physical system, we have proven the theorems under these stronger conditions.

\subsection{Completely positive maps on operator spaces}

Most of our analysis has been carried out in the Heisenberg picture. The dual maps $\TT^\ast  $, $\CC^\ast  $ and $\DD^\ast  $ have been assumed to be completely positive on the full matrix algebra. However, one could argue that only complete positivity on the operator system $\LL^\prime(\OO)$ generated by $\OO$ is required. By Arveson's extension theorem \cite[Theorem 7.5]{Paulsen2003} (or \cite[Theorem 6.2]{Paulsen2003}, since we only need the finite-dimensional version) any completely positive map $\TT^\ast  : \LL^\prime(\OO) \to \MM_D$ can be extended to a completely positive map on $\MM_D$. Hence, as long as we consider the setup relevant for quantum information, we need not distinguish whether $\TT^\ast  $ is completely positive on the full matrix algebra or on the operator system. For positive maps, this is no longer true in general (see \cite[remark after Corollary 7.6]{Paulsen2003}).

\subsection{Finitely many states}

This section will briefly address the question of what can be proven if instead of all states $\SS(\setC^D)$ we only want to measure effect operators on a subset $\SS_{\II} = \Set{\rho_i : i \in \II}$ for some states $\rho_i \in \SS(\setC^D)$ and some index set $\II \subset \setR$. We note that the situation in \Fref{sec:setup} is not changed if $\MM_D = \linspan[\setC]{\SS_{\II}}$, since again only the operator space spanned by the states matters, not the states themselves. We could thus exchange the set of all states for the set of pure states and our results in the above sections still hold.

Consider next the situation in which we allow only for states from $\SS_{\II}$, but this time we want to measure a set of effect operators $\widetilde{\OO}$ such that $\LL(\widetilde{\OO}) = \MM^\herm_D$. For example, $\widetilde{\OO} = \EE(\setC^D)$. This is the converse situation of what we considered before. Although we cannot apply the techniques used so far in this situation, this setup is actually significantly simpler. Let us adapt our definition of compressibility to this new setting.
\begin{defi}[Compression of states]
%Let $\Set{A_i: i \in [s]}$, $A_i \in \obs{\setC^D}$ for all $i \in [s]$ be a set of observables. These observables are
%Let $\LL(\OO)$ be an operator system. This operator system is 
%Let $\SS_{\II}$ be a set of states. We call $\SS_{\II}$ 
%\emph{compressible} if there exist a $d < D$, $n \in \setN$, a CPTP map $\CC: \MM_D \to \MM_{d}\otimes \setC^n$ and a CPTP map $\DD: \MM_{d}\otimes \setC^n \to \MM_D$ such that for their composition $\TT = \DD \circ \CC$, the constraints 
%\begin{equation}\label{eq:stateconstraint}
%\tr{\rho A} = \tr{\TT(\rho) A}= \tr{\rho \TT^\ast  (A)} \qquad \forall \rho \in \SS_{\II}, \forall A \in \MM^\herm_D
%\end{equation}
%are satisfied. The minimal $d$ for which this is possible is called the \emph{compression dimension}. If such a $d < D$ does not exist, $\SS_{\II}$ is said to be \emph{incompressible}.
%
Let $\SS_{\II}$ be a set of states in $\MM_D$. The  \emph{compression dimension} of $\SS_{\II}$ is the smallest $d\in\setN$ for which there is an $n\in\setN$, 
a CPTP map $\CC: \MM_D \to \MM_{d}\otimes \setC^n$ and a CPTP map $\DD: \MM_{d}\otimes \setC^n \to \MM_D$ such that for their composition $\TT = \DD \circ \CC$, the constraints 
\begin{equation}\label{eq:stateconstraint}
\tr{\rho A} = \tr{\TT(\rho) A}= \tr{\rho \TT^\ast  (A)} \qquad \forall \rho \in \SS_{\II}, \forall A \in \MM^\herm_D
\end{equation}
are satisfied. If the compression dimension equals $D$, $\SS_{\II}$ is said to be \emph{incompressible}.
\end{defi}

Then we can give a lower bound on the compression dimension in this setup.

\begin{thm}[Lower bound for states] \label{thm:nostatecomp}
Let $\SS_{\II}$ be a set of states and
\begin{equation*}
\cstar{\SS_\II} = W \left[ 0 \oplus \bigoplus_{k = 1}^s (\MM_{D^\prime_k} \otimes \idop_{m^\prime_k}) \right]W^\ast
\end{equation*}
with $D_0 + \sum_{k = 1}^s m_k D^\prime_k = D$ and $W \in \UU(D)$. Then the compression dimension is $\max_{k \in [s^\prime]} D^\prime_k$. In particular, if  $\cstar{\SS_{\II}} = \MM_D$, then $\SS_{\II}$ is incompressible.
\end{thm}

Before we can proof this, we need to prove a lemma.
\begin{lem} \label{lem:noduplicatecompression}
Let $\CC$, $\DD^\ast  : \MM_D \to \MM_{d} \otimes \setC^n$ be linear positive maps and let $\TT = \DD \circ \CC$. If the fixed point set of $\TT$ has the form $\MM_{D^\prime} \otimes \rho$, $\rho \in \SS(\setC^m)$ such that $m \cdot D^\prime = D$, then $d \geq D^\prime$.
\end{lem}
\begin{proof}
We define $\iota_{\rho}: \MM_{D^\prime} \to \MM_D$ by $\iota_{\rho}(A) = A \otimes \rho$ for all $A \in \MM_{D^\prime}$. This defines a completely positive map. We can also define $\widetilde{\TT}: \MM_{D^\prime} \to \MM_{D^\prime}$
\begin{equation*}
\widetilde{\TT} := \mathrm{Tr}_{\setC^m}{} \circ \TT \circ \iota_\rho.
\end{equation*}
Here, we have made the identification $\setC^D \simeq \setC^{D^\prime} \otimes \setC^m$. By our assumption on the fixed point set of $\TT$, we know that $\widetilde{\TT}$ is the identity map. By the same argument as in the proof of \Fref{thm:newarveson},  $d \geq D^\prime$ follows from \Fref{lem:noinfo}.
\end{proof}

\begin{proof}[Proof of \Fref{thm:nostatecomp}]
Since we require \Fref{eq:stateconstraint} to hold, we see that $\SS_\II \subset \FF_{\TT}$ for $\FF_{\TT}$ the fixed point set of $\TT$. For $\TT$ a completely positive and trace preserving map, it is known that the fixed point set has the structure
\begin{equation} \label{eq:fixedpointset}
\FF_{\TT} = U \left(0 \oplus \bigoplus_{k = 1}^s \MM_{D_k} \otimes \rho_k\right) U^\ast  
\end{equation}
with $U \in \UU(D)$, $\rho_k \in \SS(\setC^{m_k})$ and $D = D_0 + \sum_{k = 1}^s m_k D_k$, where $D_0$ is the dimension of the zero block. The $\rho_k$ can be assumed to be diagonal (we can absorb the unitaries diagonalizing them into $U$), hence 
\begin{equation*}
\cstar{\FF_{\TT}} \subset U \left(0 \oplus \bigoplus_{k = 1}^s \bigoplus_{j = 1}^{m_k} \MM_{D_k} \right) U^\ast. 
\end{equation*}
These blocks can be considered independently. Let $V_k: \setC^{D_k \cdot m_k} \hookrightarrow \setC^D$ be an isometry such that $V_k V_k^\ast  $ is the projection onto the $k$-th block in the outer direct sum. Then define $\TT_k: \MM_{m_k \cdot D_k} \to \MM_{m_k \cdot D_k}$ by
\begin{equation*}
\TT_k := \Theta_{V_k  }\circ \Theta_{U} \circ \TT \circ \Theta_{U^\ast} \circ \Theta_{V_k^\ast}.
\end{equation*}
By construction, the fixed point set of $\TT_k$ is $\MM_{D_k} \otimes \rho_k$. The map factorizes into $\CC_k: \MM_{m_k \cdot D_k} \to \MM_d \otimes \setC^n$ with 
\begin{equation*}
\CC_k = \CC \circ \Theta_{U^\ast} \circ \Theta_{V_k^\ast}
\end{equation*}
and $\DD_k:\MM_d \otimes \setC^n \to \MM_{m_k \cdot D_k}$ with
\begin{equation*}
\DD_k = \Theta_{V_k  }\circ \Theta_{U}\circ \DD.
\end{equation*}
\Fref{lem:noduplicatecompression} shows that $d \geq D_k$. Since this holds for all $k \in [s]$, it follows that $d \geq \max_{k \in [s]} D_k$. Furthermore, for all $i \in [s^\prime]$ there must be a $k \in [s]$ such that $D_k \geq D^\prime_i$, otherwise the structure of $\cstar{S_\II}$ could not be as assumed. Hence, also $d \geq \max_{k \in [s^\prime]} D^\prime_k$. That this bound is achievable can be seen through a slight modification of the construction in \Fref{thm:blockcompression}. The maps can be chosen the same (assuming $\SS_{\II}$ already to be in block-diagonal form), but the isometries need to be chosen such that they respect the block structure of the states instead of the block structure of the operators in $\OO$. Here, we treat the zero block as a direct sum of $D_0$ $1$-dimensional blocks, which do not affect the compression dimension.
\end{proof}

The same result also holds in a more general setting. For the theorem to hold, $\TT$ needs not be completely positive. Since \Fref{eq:fixedpointset} is also valid if $\TT$ is a positive, trace preserving, linear map such that the dual map satisfies the Schwarz inequality, the above theorem also holds for $\CC, \DD$ such that $\TT$ satisfies these weaker conditions (see \cite[Theorem 6.14]{Wolf2012}).

 We have shown that unlike in the converse situation, there are no redundant blocks. Whether better bounds can be shown for finite sets of both states and effect operators beyond the results from \cite{Stark2014} and \cite{Wehner2008} remains an open problem.
\section{Several copies of the same state} \label{sec:copies}

In this section, we consider the following modifications compared to \Fref{sec:setup}. We are given a set of Hermitian operators as before which we denote by $\OO$. Instead of only one state, we consider finitely many copies of the same state (provided \eg by identical preparations). Hence, we consider a quantum channel $\CC: \MM_{mD} \to \MM_{d} \otimes \setC^n$. Compression in this setting is defined as follows: 
\begin{defi}[Compression of observables using copies]
%Let $\OO \subset \MM_D^\herm$. We call $\OO$ 
%\emph{compressible} if there exist a $d < D$, a CPTP map $\CC: \MM_{mD} \to \MM_{d} \otimes \setC^n$ and a CPTP map $\DD: \MM_{d}\otimes \setC^n \to \MM_D$ such that for their composition $\TT = \DD \circ \CC$, the constraints 
%\begin{equation} \label{eq:copyconstraints}
%\Tr{\rho E} = \Tr{\TT(\rho^{\otimes m})E} = \Tr{\rho^{\otimes m}\TT^\ast  (E)} \qquad \forall \rho \in \SS(\setC^D), \forall E \in \OO
%\end{equation}
%are satisfied. The minimal $d$ for which this is possible is called the \emph{compression dimension}. If such a $d < D$ does not exist, $\OO$ is said to be %\emph{incompressible}.
%
Let $\OO$ be a set of Hermitian operators in $\MM_D$ and $m \in \setN$ the number of copies available. The  \emph{compression dimension} of $\OO$ is the smallest $d\in\setN$ for which there is an $n\in\setN$, 
a CPTP map $\CC: \MM_{mD} \to \MM_{d}\otimes \setC^n$ and a CPTP map $\DD: \MM_{d}\otimes \setC^n \to \MM_D$ such that for their composition $\TT = \DD \circ \CC$, the constraints 
\begin{equation} \label{eq:copyconstraints}
\Tr{\rho E} = \Tr{\TT(\rho^{\otimes m})E} = \Tr{\rho^{\otimes m}\TT^\ast  (E)} \qquad \forall \rho \in \SS(\setC^D), \forall E \in \OO
\end{equation}
are satisfied. If the compression dimension equals $D$, $\OO$ is said to be \emph{incompressible}.
\end{defi}
We prove now that taking copies of the state does not affect compressibility in the geometric picture.
\begin{thm}[Lower bounds on compression dimension  for finitely many copies] \label{thm:copies}
Let $\OO \subset \MM_D^\herm$ a set of Hermitian operators, $E_1$, $E_2 \in \LL(\OO)$ and 
\begin{equation*}
p(x,z) := \det[x \idop - E_1 - z E_2].
\end{equation*}
Then the smallest among the degrees of the irreducible factors of $p$ is a lower bound on the compression dimension of $\OO$. In particular, if $p$ is irreducible over the reals, then $\OO$ is incompressible.
\end{thm}
\begin{proof}
Maximizing \Fref{eq:copyconstraints} over $\rho \in \SS(\setC^D)$, we obtain
\begin{equation} \label{eq:lowerboundcopy}
\opnorm{A} = \max_{\rho \in \SS(\setC^D)} |\Tr{\rho^{\otimes m}\TT^\ast  (A)}| \qquad \forall A \in \LL(\OO).
\end{equation}
The right hand side of the above is clearly upper bounded by $\opnorm{\TT^\ast  (A)}$. Since $\TT^\ast  $ is unital, it is a contraction by the Russo-Dye theorem and
\begin{equation*}
\opnorm{\TT^\ast  (A)} \leq \opnorm{A}
\end{equation*} 
from which equality follows together with \Fref{eq:lowerboundcopy}. Thus, we are able to apply the techniques from \Fref{sec:geoargument}. Since $\CC^\ast  $ and $\DD^\ast  $  are unital as well, we obtain again
\begin{equation*}
\opnorm{E_1 + t E_2} = \opnorm{\DD^\ast  (E_1) + t \DD^\ast  (E_2)} \qquad \forall t \in \setR.
\end{equation*} 
The assertion then directly follows from \Fref{lem:nosmalleropnorm}. 
\end{proof}
Note that in this case, we have to make use of the geometric arguments since we cannot infer from \Fref{eq:copyconstraints} that $E \in \OO$ have to be fixed points of the dual channel. 
\newline
\newline
\textbf{Acknowledgements}: AB would like to thank Michael Kech for helpful discussions concerning the use of semialgebraic sets and the Stack Exchange user Simone Weil for pointing out the strategy to prove the existence of irreducible hyperbolic polynomials in any dimension. MMW thanks John Watrous for suggesting the problem and David P\'erez-Garc\'ia as well as John Watrous for interesting discussions.

\appendix

\section{Ces\`aro-mean and the support projection}

This section exposes some facts which are needed in \Fref{sec:arveson}. In the following lemma, we collect some well-known facts about the Ces\`aro-mean (\cf \cite{Lindblad1999} and \cite[Chapter 6]{Wolf2012}). We recall the definition of the transfer matrix corresponding to the projection onto the fixed points of $\RR$,
\begin{equation} \label{eq:cesaromatrix}
\hat{\RR}_\infty = \sum_{\Set{k:\lambda_k = 1}} P_{k},
\end{equation}
where $P_{\lambda_k}$ is the projection onto the (one-dimensional) Jordan block associated to the eigenvalue $\lambda_k$ of $\RR$. $\RR_{\infty}$ is the channel associated to this transfer matrix. Recall that the Ces\`aro-mean of $\RR$, if it exists, is 
\begin{equation*}
\lim_{N \to \infty} \frac{1}{N}\sum_{n = 1}^N \RR^n.
\end{equation*} 
\begin{lem} \label{lem:cesaro}
Let $\RR$ be a unital $m$-positive map on $\MM_D$ with $m \in \setN_0$. Then $\RR_\infty$ can be written as the Ces\`aro-mean of $\RR$, it is unital and $m$-positive, $\RR_\infty$ is idempotent, $\RR_\infty \circ \RR = \RR_\infty = \RR \circ \RR_\infty $ and $\RR_\infty$ has the same fixed point set as $\RR$.
\end{lem}

\begin{proof}
The spectral radius of $\RR$ is equal to 1 by \cite[Proposition 6.1]{Wolf2012}. Note furthermore that \cite[Proposition 6.2]{Wolf2012} implies that the Jordan blocks belonging to eigenvalues of modulus $1$ are one-dimensional. By the same argument as in \cite[Proposition 6.3]{Wolf2012}, the first assertion then follows. From there, unitality and $m$-positivity directly follow.
Looking at its transfer matrix, $\RR_\infty$ is clearly idempotent, \ie
\begin{equation*}
\RR_\infty \circ \RR_\infty = \RR_\infty.
\end{equation*}
$\RR_\infty = \RR \circ \RR_\infty$ holds since for every $A$ in the range of $\RR_\infty$, we know that $\RR(A) = A$. Furthermore,
\begin{equation} \label{eq:Tinv}
\RR_\infty \circ \RR = \RR_\infty
\end{equation}
follows by multiplication of the respective transfer matrices and using that Jordan blocks for eigenvalues of modulus $1$ are one-dimensional.
%since for $A \in \MM_D$
%\begin{align*}
%&\norm{\frac{1}{N}\sum_{n = 1}^N \RR^n(\RR(A)) - \frac{1}{M}\sum_{m = 1}^M \RR^m(A)}_\infty\\
%\leq & \norm{\frac{1}{N}\sum_{n = 1}^N \RR^n(A) - \frac{1}{M}\sum_{m = 1}^M \RR^m(A)}_\infty+\frac{1}{N}\norm{\RR^{N+1}(A) - \RR(A)}_\infty\\
%\leq & \norm{\frac{1}{N}\sum_{n = 1}^N \RR^n(A) - \frac{1}{M}\sum_{m = 1}^M \RR^m(A)}_\infty+\frac{2\norm{A}_\infty}{N}.
%\end{align*}
%This goes to zero as $N,M \to \infty$. We used here that $\norm{\RR} \leq 1$. 
Obviously, for $A \in \MM_D$ such that $\RR(A) = A$, also $\RR_\infty(A) = A$ holds; therefore, the fixed point sets are equal by the definition of $\RR_\infty$.
\end{proof}

We also need the fact that the Ces\`aro-mean of a Schwarz map is again a Schwarz map. To prove this, we will need a lemma.

\begin{lem} \label{lem:schwarzhelp}
Let $\TT, \RR$ be two Schwarz maps on $\MM_D$. Then $\TT \circ \RR$ is a Schwarz map as well. Furthermore $\lambda \TT + (1-\lambda) \id$ is a Schwarz map for all $\lambda \in [0,1]$. 
\end{lem}
\begin{proof}
Applying the Schwarz inequality twice, we obtain
\begin{equation*}
(\TT \circ \RR)(A) (\TT \circ \RR)(A^\ast) \leq \TT(\RR(A)\RR(A)^\ast) \leq \TT \circ \RR (A A^\ast),
\end{equation*}
where we used positivity of both $\TT$ and $\RR$. For the second assertion, we compute
\begin{align*}
&(\lambda \TT + (1-\lambda) \id)(A A^\ast) - (\lambda \TT + (1-\lambda) \id)(A)(\lambda \TT + (1-\lambda) \id)(A^\ast) \\
 &=  \lambda (1 - \lambda) \left[\TT(A A^\ast) + A A^\ast -  A \TT(A^\ast) - \TT(A) A^\ast \right] 
+ \lambda^2 \left[\TT(A A^\ast) - \TT(A) \TT(A^\ast)\right].
\end{align*}
We have to show that the above expression is positive. The second term is positive by the Schwarz inequality. We can reformulate the first term as
\begin{align*}
\TT(A A^\ast) + A A^\ast -  A \TT(A^\ast) - \TT(A) A^\ast = \begin{bmatrix}
\idop & A
\end{bmatrix} \begin{bmatrix}
\TT(A A^\ast) & - \TT(A) \\ -\TT(A^\ast) & \idop 
\end{bmatrix} \begin{bmatrix}
\idop \\ A^\ast
\end{bmatrix}.
\end{align*}
The operator matrix can be shown to be positive semidefinite using the Schur complement, so the right hand side of the above is positive as well.
\end{proof}

\begin{lem} \label{lem:schwarzcesaro}
Let $\TT: \MM_D \to \MM_D$ be a Schwarz map. Then the Ces\`aro-mean of $\TT$ is a Schwarz map as well.
\end{lem}
\begin{proof}
The statement for 
\begin{equation*}
\frac{1}{N} \sum_{n = 1}^N \TT^n
\end{equation*}
follows by induction. Using \Fref{lem:schwarzhelp}, we infer that
\begin{equation*}
\frac{1}{2} \left[ \TT^2 + \TT\right] = \frac{1}{2} \left[ \TT + \id\right] \circ \TT
\end{equation*}
is a Schwarz map. By the same lemma, it follows that 
\begin{equation*}
\frac{1}{N+1} \sum_{n = 1}^{N+1} \TT^n = \left(\frac{1}{N+1} \id + \left( 1 - \frac{1}{N+1} \right)\frac{1}{N} \sum_{n = 1}^N \TT^n\right) \circ \TT
\end{equation*}
is a Schwarz map using the induction hypothesis. The statement follows taking the limit $N \to \infty$.
\end{proof}

The rest of this section focuses on the support projection. We are only concerned with matrix algebras, so we assume $\A \subset \MM_D$ to be a finite-dimensional C$^\ast  $-algebra and let $\A_+$ denote the positive elements in this algebra. For the case of von Neumann algebras of arbitrary dimensions see \eg \cite{Dixmier1957} or \cite[III.2.2.25]{Blackadar2006}. The support projection of a Schwarz map is not to be confused with the support projection of its transfer matrix. They are in general not the same. First we define the set 
\begin{equation*}
\NN = \Set{A \in \A : \RR(A^\ast  A) = 0}
\end{equation*}
for some Schwarz map $\RR: \A \to \A$. This set contains projections, as we shall see. Using the spectral decomposition, we may write $A^\ast   A = \sum_{i = 1}^n\sigma_i^2 P_i$, where $\sigma_i > 0 $, $i \in [n]$ are the distinct singular values of $A$ and $P_i \in \A$ the corresponding spectral projections. Then 
\begin{equation*}
\RR(A^\ast   A) = \sum_{i = 1}^n \sigma_i^2 \RR(P_i).
\end{equation*}
The sum is zero if and only if $\RR(P_i) = 0$ for all $i \in [n]$ and thus also the support projection of $A$, $V_A = \sum_{i=1}^n{P_i}$, is in $\NN$. By the lattice structure of the set of projections, there is a unique maximal projection in $\NN$. We will denote this projection by $Q$.
%This can be rewritten as 
%\begin{equation*}
%\NN = \Set{A \in \A : \RR(B^\ast   A) = 0 \quad \forall B \in \A}
%\end{equation*}
%and can therefore be seen to be a closed left ideal. The rewriting is possible, since by the Cauchy-Schwarz inequality for 2-positive maps
%\begin{equation*}
%\norm{\RR(B^\ast  A)}^2 \leq \norm{\RR(B^\ast  B)}\norm{\RR(A^\ast  A)}.
%\end{equation*}
%From the theory of von Neumann algebras \cite[p.45, Corr. 3]{Dixmier1957}, there is a unique maximal projection $Q \in \NN$ in this ideal such that $\NN = \Set{A \in \A : AQ = A}$.
%Explicitly, we have the following
%\begin{cor}
%Let $\A$ be a von Neumann algebra, $\NN$ a left ideal closed in the weak-$\ast  $-topology (also known as ultraweak topology). Then there exists a unique projection $Q$ in $\A$ such that $\NN$ is the set of all $A \in \A$ which satisfy $A = AQ$.
%\end{cor}
Using the existence of such a $Q$, we get 
\begin{align*}
\opnorm{\RR(AQ)}^2 &\leq \opnorm{\RR(QA^\ast   AQ)} = 0\qquad \forall A \in \A,\\
\opnorm{\RR(QA)}^2 &\leq \opnorm{\RR(QAA^\ast   Q)} = 0\qquad \forall A \in \A,
\end{align*}
where we used the C$^\ast  $-property, the fact that positive maps are hermiticity preserving, the Schwarz inequality \Fref{eq:schwarz} and 
\begin{equation*}
\opnorm{\RR(Q B Q)} \leq \opnorm{B} \opnorm{\RR(Q)} = 0 \qquad \forall B \in \A_+.
\end{equation*}

 This implies
\begin{equation*}
\RR(AQ) = \RR(QA) = 0\qquad \forall A \in \A.
\end{equation*}
Hence we can define the support projection as $P := \idop - Q$. By the above property of $Q$, it fulfills 
\begin{equation*}
\RR(A) = \RR(PA) = \RR(AP) = \RR(PAP) \qquad \forall A \in \A.
\end{equation*}
The following lemma collects the properties of the support projection which we use.
\begin{lem} \label{lem:propsuppproj}
Let $\RR: \A \to \A$ be a Schwarz map. Then for its support projection $P$, we have that  
\begin{equation*}
\RR(A) = \RR(PA) = \RR(AP) = \RR(PAP) \qquad \forall A \in \A
\end{equation*}
and $\RR|_{P \A P}$ is faithful.
\end{lem}
\begin{proof}
We only need to check the last claim, since we have already shown the rest.
Being faithful on $P\A P$ means that the implication
\begin{equation} \label{eq:cpfaithful}
\RR(A) = 0 \quad \rightarrow \quad PAP = 0 \qquad \forall A \in \A_+
\end{equation}
is true. This can be seen to hold as follows: Assume $\RR(B) = 0$ for some $B \in \A_+$. Then there is an $A \in \A$ such that $B = A^\ast  A$, because $B$ is positive. For this $A$ we know that $A \in \NN$ and $QA^\ast   AQ = A^\ast   A$ by the definition of $Q$. Hence also 
\begin{equation*}
B = A^\ast  A = QA^\ast  AQ = QBQ.
\end{equation*}
However, as $P = \idop - Q$, this gives 
\begin{equation*}
PBP = 0
\end{equation*}
as claimed.
\end{proof}
Note that for general (non-positive) $B \in \A$, the implication in \Fref{eq:cpfaithful} is no longer true.

\section{Existence of both irreducible and hyperbolic polynomials of any degree} \label{sec:polyappendix}

The aim of this subsection is to show that there exist homogeneous polynomials of any degree which are both irreducible and hyperbolic. This was is used in \Fref{sec:geoargument}. This well-known fact from algebraic geometry will be proven here for convenience. It is clear that there are irreducible polynomials of any degree since $p(x,y,z) = x^d + y^d - z^d$ is irreducible for any $d \in \setN$. Furthermore, it has been shown in \cite{Nuij1968} that the set of polynomials $p \in \bfH^d(n)$ hyperbolic with respect to a fixed point $e \in \setR^n$ has non-empty interior in $\bfH^d(n)$ (see also \cite[Theorem 2.1]{Gueler1997}). It is, however, not clear a priori that there are elements which fulfill both properties, since the $p(x,y,z)$ given above are not hyperbolic, as can be checked easily. The idea now is to prove that the set of reducible polynomials in $\bfH^d(n)$ does not contain any open subset, which would then mean that the set of irreducible and the set of hyperbolic elements in this space have non-empty intersection. The argument proceeds by dimension counting. We restrict to the case $n = 3$ for simplicity. Since we will be interested in normalized polynomials (\ie $p(e) = 1$ for $e = (1, 0,  0)$), let $\bfH^d_N(3)\subset \bfH^d(3)$ be the affine subspace of such polynomials, where normalization decreases the dimension by one. We will identify $\bfH^d_N(3) \simeq \setR^{\dim\bfH^d(3)-1}$, since we are only interested in the topology and measure on this affine space. Redoing the argument by Nuij shows that the set of normalized hyperbolic polynomials has non-empty interior in $\bfH^d_N(3)$ as well, since it basically only uses that the simple roots of an univariate polynomial depend continuously on the coefficients of the polynomial. 

\begin{lem} \label{lem:reduciblenoopensets}
The set of reducible elements over the reals in $\bfH_N^d(3)$, $d \in \setN$, $d > 2$ does not contain any subset which is open in Euclidean topology. Moreover, this set has Lebesgue measure zero.
\end{lem}
\begin{proof}
Let $p \in \bfH_N^d(3)$ be a reducible element. Then there are $q \in \bfH_N^k(3)$, $r \in \bfH_N^{d-k}(3)$, $k \in [d-1]$, such that $p = q \cdot r$. The fact that these polynomials can be chosen normalized follows since for $q(e) = c \neq 0$, necessarily $r(e) = 1/c$ by normalization of $p$ and the polynomials can be multiplied by $c$ and $1/c$, respectively, to obtain a decomposition into normalized elements. Hence, we define a mapping 
\begin{align*}
\Phi:~& \bfH_N^k(3) \times  \bfH_N^{d-k}(3)  \to \bfH_N^d(3) \\
&(q  ,  r) \mapsto  q \cdot r.
\end{align*}
$\bfH_N^k(3)$ is a semialgebraic set (e.g. as $\ZZ(x_0 - 1)$, where $x_0$ is the coefficient belonging to $x^k$) with dimension $\binom{3 + k -1}{k}-1$ (by \cite[Proposition 2.8.1]{Bochnak1998}, since dividing out the ideal $(x_0 -1)$ decreases the dimension by one). Moreover, $\Phi$ is a semi-algebraic mapping (see \cite[Definition 2.2.5]{Bochnak1998}), since its graph can be expressed as
\begin{align*}
\Bigg\{&(q,r,p) \in \setR^{\dim(\bfH_N^k(3))}\times \setR^{\dim(\bfH_N^{d-k}(3))} \times \setR^{\dim(\bfH_N^d(3))} : \sum_{\substack{i_l + j_l = m_l \\ l \in [3]}} q_{i} r_{j} - p_{m} = 0;  \\ & m_1 \in [d-1],~m_2, m_3 \in [d],~m_1 + m_2 + m_3 = d\Bigg\}.
\end{align*}
Here, $i,j,m \in \setN_0^3$ are multi-indices such that $|i| = k$, $|j| = d - k$ and $|m| = d$. This is a finite collection of polynomial equalities which have to be fulfilled, thus it is a semi-algebraic set. We have written $q_{i}$ to be the coefficient belonging to $x^{i_1} y^{i_2} z^{i_3}$ of the polynomial $q$ for clarity (same for $p$, $r$). Note that $q_{(k,0,0)} = r_{(d-k,0,0)} = p_{(d,0,0)} = 1$ has been fixed beforehand by normalization. By \cite[Proposition 2.2.7]{Bochnak1998}, we know that the image of $\Phi$ for a fixed $k$ is also a semi-algebraic set, likewise this holds for the set of reducible elements in $\bfH_N^d(3)$, since it is a finite union of semi-algebraic sets. Now we come back to the dimensions of the sets involved. By \cite[Proposition 2.8.5 (ii)]{Bochnak1998}, the domain of $\Phi$ has dimension
\begin{equation*}
\binom{3+k-1}{k} + \binom{3 + d-k -1}{d-k}-2.
\end{equation*}
Further,
\begin{equation*}
\binom{3 + d-1}{d} - \binom{3+k-1}{k} - \binom{3 + d-k -1}{d-k} + 1 = (d - k) k,
\end{equation*}
which is greater equal $ d - 1$ for $k \in [d-1]$ and hence strictly positive for $d > 1$. Hence the set of reducible elements has  dimension strictly smaller than the dimension of $\bfH_N^d(3)$, $d > 1$, by \cite[Proposition 2.8.5 (i)]{Bochnak1998} and \cite[Proposition 2.8.8]{Bochnak1998}. This implies that it cannot contain any open $U \subset \bfH_N^d(3)$, since $\II(U) = \Set{0}$ necessarily, but there is at least one non-trivial polynomial vanishing on the set of reducible elements (otherwise this set would have full dimension by \cite[Definition 2.8.1]{Bochnak1998}), which would also vanish on any subset of these. By \cite[Proposition A.1]{Kech2017}, any semi-algebraic $\BB \subset \setR^m$ of dimension less than $m$ has zero $m$-dimensional Hausdorff measure and hence also zero Lebesgue measure, since those only differ by a constant factor on $\setR^m$. 
\end{proof}

\section{Matrix computations}

This sections contain some elementary computations needed to show the irreducibility of the polynomial in \Fref{prop:irredexample}. Let the matrices $\tilde{A}(x)$, $B \in \MM_D$, $D \geq 2$, be defined as follows:
\begin{equation*}
\tilde{A}_{kl}(x) = \begin{cases} 0 & k < l \\ x & k = l \\ 1 & k > l \end{cases} \qquad B_{kl} = \frac{1}{2}\begin{cases} \imI & k < l \\ 0 & k = l \\ -\imI & k > l \end{cases} \qquad k,l \in [D].
\end{equation*}
The aim of this section is to compute $\det[\tilde{A}(x) + \epsilon B]$ up to first order in $\epsilon$ and to show that the first order term has to vanish at $x = 0$ under some assumptions. On the way, we need to prove several lemmas which are of little interest in themselves. The first is the following sum formula which we will use several times:
\begin{lem} \label{lem:sum}
The following identity is true for $k \in \setN_0$:
\begin{equation*}
\sum_{j = 0}^k \left[ - \frac{(x-1)^{j}}{x^{j+2}}\right] + \frac{1}{x} = \frac{(x-1)^{k+1}}{x^{k+2}}.
\end{equation*}
\end{lem}
\begin{proof}
%We note that
%\begin{equation*}
%\frac{(x-1)^j}{x^{j+1}} - \frac{(x-1)^j}{x^{j+2}} = \frac{(x-1)^{j+1}}{x^{j+2}}.
%\end{equation*}
The statement follows by induction. %For $k = 0$,
%\begin{equation*}
%-\frac{1}{x^2} + \frac{1}{x} = \frac{x-1}{x^2}.
%\end{equation*}
%For the induction step $k \to k+1$,
%\begin{equation*}
%\sum_{j = 0}^{k+1} \left[ - \frac{(x-1)^{j}}{x^{j+2}}\right] + \frac{1}{x} =- \frac{(x-1)^{k+1}}{x^{k+3}} + \frac{(x-1)^{k+1}}{x^{k+2}} = \frac{(x-1)^{k+2}}{x^{k+3}}.
%\end{equation*}
%by the equality given initially.
\end{proof}
We want to give the inverse of $\tilde{A}(x)$ provided $x \neq 0$.
\begin{lem}
Assume that $x \neq 0$. Then the inverse of $\tilde{A}(x)$ is given by
\begin{equation*}
C_{kl}(x) := \begin{cases} 0 & k < l \\ \frac{1}{x} & k = l \\ - \frac{(x-1)^{k-l-1}}{x^{k-l+1}} & k > l \end{cases}\qquad k,l \in [D].
\end{equation*}
\end{lem}
\begin{proof}
First, note that $\tilde{A}(x)$ is invertible for $x \neq 0$, since $\det[\tilde{A}(x)] = x^D$. We want to show that $C(x) \tilde{A}(x) = \idop$. For now, let $C(x) \tilde{A}(x) =: F$. Note that $F$ is lower triangular since $C(x)$ and $\tilde{A}(x)$ are. For $i \geq j$, we find
\begin{equation*}
F_{ij} = \sum_{k = j}^i C_{ik}(x)\tilde{A}_{kj}(x).
\end{equation*}
For $i = j$, we have $F_{ii} = \frac{1}{x} x = 1$. For $i > j$, we obtain
\begin{align*}
F_{ij} &= -\frac{(x-1)^{i-j-1}}{x^{i-j+1}}x + \sum_{k = j+1}^{i-1} \left[-\frac{(x-1)^{i-k-1}}{x^{i-k+1}}\right] + \frac{1}{x} \\
& = -\frac{(x-1)^{i-j-1}}{x^{i-j}}+ \sum_{k = 0}^{i-j-2} \left[-\frac{(x-1)^{k}}{x^{k+2}}\right] + \frac{1}{x} \\
& = 0.
\end{align*}
The last equality follows by \Fref{lem:sum}. Hence $F_{ij} = \delta_{ij}$.
\end{proof}
This can be used to compute the trace of $\tilde{A}^{-1}(x)B$.
\begin{lem} \label{lem:invtrace}
Let $x \neq 0$. Then 
\begin{equation*}
\tr{\tilde{A}^{-1}(x)B} = \frac{\imI}{2} \left[1 - \frac{D}{x} - \frac{(x-1)^D}{x^D} \right].
\end{equation*}
\end{lem}
\begin{proof}
We first need to compute the diagonal entries of $\tilde{A}^{-1}(x)B$. Since $B_{ii} = 0$ $\forall i \in [D]$ and $\tilde{A}^{-1}(x)$ is lower triangular, we have
\begin{align*}
\left[\tilde{A}^{-1}(x)B\right]_{jj} &= \sum_{k = 1}^{j - 1} \tilde{A}^{-1}_{jk}(x) B_{kj} = \frac{1}{2}\sum_{k = 1}^{j - 1} \left[ - \frac{\imI (x-1)^{j-k-1}}{x^{j-k+1}}\right] \\&= \frac{1}{2}\sum_{k = 0}^{j - 2} \left[ - \frac{\imI (x-1)^{k}}{x^{k+2}}\right].
\end{align*}
Taking the trace of this, we obtain
\begin{align*}
\tr{\tilde{A}^{-1}(x)B}& = \frac{1}{2}\sum_{j = 2}^D \sum_{k = 0}^{j - 2} \left[ - \frac{\imI (x-1)^{k}}{x^{k+2}}\right] \\
& = \frac{\imI}{2} \sum_{j = 0}^{D-2} \left[\frac{ (x-1)^{j+1}}{x^{j+2}} - \frac{1}{x}\right] \\
%& = -\frac{\imI}{2} (x-1) \sum_{j = 0}^{d-2} \left[-\frac{ (x-1)^{j}}{x^{j+2}}\right] - (d-1)\frac{\imI}{x} \\
%& = -\frac{\imI}{2} \frac{ (x-1)^{d}}{x^{d}} + \frac{\imI}{2}\frac{x-1}{x} - (d-1)\frac{\imI}{2}{x} \\
& = \frac{\imI}{2} -\frac{\imI}{2} \frac{ (x-1)^{D}}{x^{D}} - D\frac{\imI}{2x},
\end{align*}
where we have used \Fref{lem:sum} in the second and third equality.
\end{proof}
Finally, we can use these computations to expand $\tilde{A}(x) + \epsilon B$ to first order in $\epsilon$.

\begin{lem} \label{lem:detexpansion}
We can expand the determinant of $\tilde{A}(x) + \epsilon B$ in terms of $\epsilon$ as
\begin{equation*}
\det[\tilde{A}(x) + \epsilon B] = x^D -  \frac{\epsilon \imI}{2} \left[D x^{D-1} + (x-1)^D - x^D\right] + \OO(|\epsilon|^2).
\end{equation*}
\end{lem}
\begin{proof}
Let $x \neq 0$. Let $f: \setC \to \setC$, $f(\epsilon) = \det[\tilde{A}(x) + \epsilon B]$. By Taylor's theorem, we have
\begin{equation*}
f(\epsilon) = f(0) + f^\prime(0)\epsilon + \OO(|\epsilon|^2).
\end{equation*}
By Jacobi's formula, it follows that 
\begin{equation*}
\frac{d}{d t} \det[\tilde{A}(x) + t B)] \Big|_{t = 0} = \tr{\mathrm{adj}(\tilde{A}(x))B}.
\end{equation*}
Using that $\tilde{A}(x)\mathrm{adj}(\tilde{A}(x)) = \mathrm{det[}\tilde{A}(x)] \idop$ by the definition of the adjugate matrix, we infer
\begin{equation*}
\det[\tilde{A}(x) + \epsilon B] = \det[\tilde{A}(x)] + \epsilon \det[\tilde{A}(x)]\tr{\tilde{A}^{-1}(x)B} + \OO(|\epsilon|^2).
\end{equation*}
%Then
%\begin{equation*}
%\det[\tilde{A}(x) + \epsilon B)] = \sum_{k = 0}^d \epsilon^k \tr{\mathrm{adj}_k(\tilde{A}(x)) C_k(B)},
%\end{equation*}
%where $C_k(B)$ is the $k$-th compound  matrix of $B$ and $\mathrm{adj_k}(\tilde{A}(x))$ the $k$-th adjugate matrix of $\tilde{A}(x)$ (\cf \cite[Section 0.8.12]{Horn2012} and errata for p. 29 of this book). Using $\mathrm{adj_0}(\tilde{A}(x)) = \det[\tilde{A}(x)]$, $\mathrm{adj_1}(\tilde{A}(x)) = \mathrm{adj}(\tilde{A}(x)) = \det[\tilde{A}(x)]\tilde{A}^{-1}(x)$ and $C_0(B) = 1$, $C_1(B) = B$, we obtain
%\begin{equation*}
%\det[\tilde{A}(x) + \epsilon B] = \det[\tilde{A}(x)] + \epsilon \det[\tilde{A}(x)]\tr{\tilde{A}^{-1}(x)B} + \LL(\OO)(\epsilon^2).
%\end{equation*}
By \Fref{lem:invtrace} and $\det[\tilde{A}(x)] = x^D$, the statement follows for $x \neq 0$. The result extends to $x = 0$ by continuity.
\end{proof}

\begin{lem} \label{lem:constantzero}
Let $p$, $q$, $r \in \setC[x, y]$ such that $q(0,0) = 0 = r(0,0)$ and $p = q \cdot r$. Let the expansion in $\epsilon$ of $p(x,\epsilon)$ be 
\begin{equation*}
p(x,\epsilon) = \sum_{k = 0}^D p_k(x) \epsilon^k
\end{equation*}
for $D$ the degree of $p$ in $y$ and $p_k(x) \in \setC[x]$ for all $k \in [D]\cup \Set{0}$. Then $x = 0$ is a root of $p_1(x)$. 
\end{lem}
\begin{proof}
By the above expansion, we can write
\begin{equation*}
p_1(x) = \frac{d}{d\epsilon}p(x,\epsilon)\Big|_{\epsilon = 0}.
\end{equation*}
Using the definition of $p$, we obtain
\begin{equation*}
p_1(0) =\frac{d}{d\epsilon} q(0,\epsilon)\Big|_{\epsilon = 0} r(0,0) + q(0,0) \frac{d}{d\epsilon} r(0,\epsilon)\Big|_{\epsilon = 0}, 
\end{equation*}
which is zero since we assumed $q(0,0) = 0 = r(0,0)$.
%We write
%\begin{equation*}
%p(x,y) = \sum_{i = 0}^d \sum_{j = 0}^{d-i} p_{ij} x^i y^j
%\end{equation*}
%with $p_{ij} \in \setC$, $\forall j \in [d-i]$, $\forall i \in [d]$ and likewise for $q$, $r$. Then $q_{00} = 0 = r_{00}$ and for the terms which are linear in $\epsilon$
%\begin{equation*}
%\sum_{i = 0}^d x^i p_{i1} = \left(\sum_{i = 0}^k x^i q_{i1}\right) \left(\sum_{j = 1}^{d-k} x^j r_{j0}\right) + \left(\sum_{i = 1}^k x^i q_{i0}\right) \left(\sum_{j = 0}^{d-k} x^j r_{j1}\right)
%\end{equation*}
%with $d-1 \geq k \geq 1$, $k \in \setN$. Since every term on the right hand side has at least a factor $x$, the claim follows.
\end{proof}

\section{No information without disturbance for positive maps}

In the ordinary setting, the statement that there is no information without disturbance is proven for completely positive maps, because those are the physically relevant evolutions of the system. The statement then has a short proof using Choi matrices. In this section, we show that the statement still holds for merely positive maps. This is used \eg in \Fref{sec:generalizations}.

\begin{lem} \label{lem:noinfo}
Let $\TT_i: \MM_D \to \MM_D$, $i \in [s]$, be a collection of positive linear maps such that
\begin{equation*}
\sum_{i = 1}^s \TT_i = \id.
\end{equation*}
Then $\TT_i = c_i \id$ for some $c_i \geq 0$ for all $i \in [s]$ and $\sum_{i = 1}^s c_i = 1$. 
\end{lem}
\begin{proof}
Let $\dyad{\psi} \in \SS(\setC^D)$. Then $\TT_i(\dyad{\psi}) = c_i(\psi, \psi) \dyad{\psi}$ for some number $c_i(\psi,\psi)\geq 0$ and for all $i \in [s]$. This follows from positivity of the maps, the fact that they sum to the identity and because the rank one projections are extremal in the set of states. We have to show that the constant does not depend on the state. Consider $A := (x\Ket{\psi} + y\Ket{\phi})(\bar{x}\Bra{\psi} + \bar{y}\Bra{\phi})$, with $\Ket{\psi}$, $\Ket{\phi}$ orthonormal, $x$, $y \in \setC$. Again, $\TT_i(A) = c_i(A) A$. By linearity, 
\begin{equation*}
\TT_i(A) = c_i(\psi, \psi)|x|^2\dyad{\psi} + c_i(\phi, \phi)|y|^2\dyad{\phi} + \TT_i(x\dyad{\psi}{\phi}\bar{y} + y\dyad{\phi}{\psi}\bar{x}).
\end{equation*}
Let $\tilde{A} = a\dyad{\psi} + b \dyad{\phi} + c\dyad{\phi}{\psi}+ \bar{c}\dyad{\psi}{\phi}$. This matrix is positive semidefinite for $a \geq 0$, $ab - |c|^2 \geq 0$, $b \geq 0$, $c \in \setC$. Note that we can scale $a \to \lambda a$, $b \to \frac{1}{\lambda}b$ for $\lambda \in \setR\setminus \Set{0}$ while keeping $c$ constant. With $\TT_i(\tilde{A}) \leq \tilde{A}$ and $\TT_i(\tilde{A}) \geq 0$, we can infer 
\begin{equation*}
\Bra{\theta}\TT_i(x\dyad{\psi}{\phi}\bar{y} + y\dyad{\phi}{\psi}\bar{x})\Ket{\theta} = 0 \qquad \forall \Ket{\theta} \in \Set{\Ket{\psi}, \Ket{\phi}}
\end{equation*}
by scaling with an appropriate $\lambda$. Hence, 
\begin{equation*}
\TT_i(x\dyad{\psi}{\phi}\bar{y} + y\dyad{\phi}{\psi}\bar{x}) = c_i(x\psi, y\phi) \dyad{\psi}{\phi} + \bar{c_i}(x\psi, y\phi)\dyad{\phi}{\psi}.
\end{equation*}
Thus, computing $\Bra{\theta_1}\TT_i(A)\Ket{\theta_2}$ for $\Ket{\theta_1}$, $\Ket{\theta_2} \in \Set{\Ket{\psi}, \Ket{\phi}}$ yields that both $c_i(\psi, \psi) = c_i(\phi, \phi)$ and $c_i(x\psi, y\phi) =  c_i(\psi,\psi)x \bar{y}$. Thus, the constants do not depend on $\Ket{\psi}$ and $\Ket{\phi}$. Choosing an orthonormal basis and the corresponding usual basis of Hermitian operators, this implies that $\TT_i = c_i \id$ for $c_i \geq 0$ for all $i \in [s]$.
\end{proof}

\section{Complexity of block diagonalization and InterpolationSDP} \label{sec:complexity}
\subsection{Block diagonalization}
In this section, we will analyze the complexity of determining the minimal compression dimension. This is needed in \Fref{sec:computing}. We start with the block diagonalization part. Assume we are given linearly independent Hermitian operators $\Set{E_1, \ldots, E_k} = \OO \subset \MM_D$ with entries in $\setQ(\imI)$ (the complex numbers with rational real and imaginary part). We first need to determine the composition of the C$^\ast  $-algebra generated by $\OO$ into irreducible components,
\begin{equation} \label{eq:irreddecomp}
\cstar{\OO} = U^\ast   \left(\bigoplus_{i = 1}^s \MM_{D_i} \otimes \idop_{m_i}\right) U,
\end{equation}
where $U \in \MM_D$ is unitary and $m_i \in \setN$ for all $i \in [s]$. For this, we use the complex version of the algorithm proposed in \cite{Murota2010}. This algorithm is formulated over the real numbers and can be adapted to the complex case. However, we will show that we can still use it if we only allow for algebraic numbers, which we will denote by $\setA$. This is more realistic for practical computations. For the real algebraic numbers, we will write $\setA_{\setR}$. By \Fref{eq:irreddecomp}, we can write
\begin{equation*}
E_j = U^\ast   \left( \bigoplus_{i = 1}^s E_j^i  \otimes \idop_{m_i}\right) U \qquad \forall j \in [k]
\end{equation*}
with $E_j^i \in \MM_{D_i}$ for all $i \in [s]$, $j \in [k]$. We note that the entries of $U$, $E_j^i$ can be chosen to be in $\setA$, since they are the solutions to a system of polynomial equalities, which we can split into real and imaginary part. As $\setA_\setR$ is a real closed field, it follows by the Tarski transfer principle \cite[11.2.3]{Marshall2008} that this system of equations has a solution in $\setA_{\setR}$ if and only if it has a solution in $\setR$. The latter is guaranteed by \Fref{eq:irreddecomp}.

If the $E_1, \ldots, E_k$ do not linearly span $\cstar{\OO}$, we may find such a basis by a procedure similar to the one described in \cite[comment after Proposition 5]{Murota2010}. Note that in the complex case, we need to add elements of the form $\imI(AB-BA)/2$ in each step as well. We further assume that we are given a finite set $\BB \subset \setQ$ with at least $(s / \epsilon) \max_{i \in [s]} D_i$ elements for some $\epsilon \in (0,1)$. Choosing $r \in \BB^k$ randomly from a uniform distribution, the element $E(r) = r_1 E_1 + \ldots r_k E_k$ is generic with probability at least $1-\epsilon$. Generic means that elements in the different simple components in \Fref{eq:irreddecomp} have different eigenvalues. This can be guaranteed by avoiding the zero set of a polynomial which is the product of the resultants of the characteristic polynomials for the respective blocks. The lower bound follows from the Schwartz-Zippel lemma \cite[Corollary 1]{Schwartz1980} applied to that polynomial and a union bound. See \cite[Proposition 3]{Murota2010} for details. We can rescale $r$ such that $r \in \setQ$, $\norm{r}_2 \leq 1$. Then we can compute the characteristic polynomial of $E(r)$ and use the (probabilistic) factorization algorithm based on basis reduction in \cite[Corollary 16.25]{Gathen2013} to factor it into irreducible components. Note that this gives us the eigenvalues of $E(r)$ with their respective algebraic multiplicities, since the algebraic numbers are defined by their minimal polynomials. Using Gaussian elimination, we can obtain the corresponding eigenvectors. Grouping the eigenvalues into sets as described in \cite[Proposition 2]{Murota2010}, we have found the decomposition into irreducible elements. The second part of the algorithm, finding the irreducible factors, can be carried out exactly as described by \cite{Murota2010}. The overall complexity is dominated by the factorization of the characteristic polynomial. Its maximal coefficient has modulus at most $(k D M)^D$, where $M = \max_{i \in k}\opnorm{E_i}$. Therefore, the factorization needs an expected number of $\OO(D^{10} \mathrm{polylog}(k)\mathrm{polylog}(D)\mathrm{polylog}(M))$ arithmetic operations. The algorithm succeeds with probability at least $1 - \epsilon$, because the element $E(r)$ needs to be generic. Thus we have obtained $E_1, \ldots, E_k$ in block diagonal form as required for \Fref{alg:mindimalgo}.
\subsection{Complexity of InterpolationSDP}

In the rest of this section, we will comment on the complexity of solving the semidefinite program $\mathrm{InterpolationSDP}(E_1, \ldots, E_k)(j)$. We have to convert the unbounded optimization problem into a feasibility problem to be able to practically solve it. The new SDP is:

\emph{Given $E_i^j \in \MM^\herm_{D_j}$, $j \in [s]$, $i \in [k]$, determine whether there are $H_i \in \MM_{D_1}$, $i \in [k+1]$ such that} 
\begin{align*}
& \sum_{i = 1}^k \begin{bmatrix}
0 & & & \\ &E_i^2 & & \\ &  &\ddots & \\ & & & E_i^s 
\end{bmatrix} \otimes H_i \geq 0\\
& - 1 - \left(\sum_{i = 1}^k \tr{(E^1_i)^T H_i}\right) \geq 0.
\end{align*}
\begin{prop}
If the $E_i^j$ have entries in $\setQ(\imI)$ for all $j \in [s]$, $i \in [k]$, then the feasibility of this SDP can be determined in $\OO(kD_1^6 D^4) + (D_1 D)^{\OO(k D_1^2)}$ operations.
\end{prop}
\begin{proof}
This follows from the results in \cite{Porkolab1997}. Theorem 5.7 of that paper states that the given symmetric $n \times n$ matrices $Q_0, \ldots, Q_m$ with integer entries, the question whether there are real numbers $x_1, \ldots, x_m$ such that
\begin{equation*}
Q_0 + x_1 Q_1+ \ldots + x_m Q_m \geq 0
\end{equation*} 
can be decided using $\OO(mn^4) + n^{\OO(\min\Set{m, n^2})}$ operations. To use this theorem, we have to convert the SDP into standard form. This can be done using a basis of the Hermitian matrices and expressing the $H_i$ as a real combination of basis elements. The two constraints can be combined into one writing them as a block matrix. Multiplying the equations by an appropriate positive integer, we can assume that they have integer coefficients. Finally, we can convert a complex SDP into a real SDP while increasing the dimension of the matrices by a factor of $2$. Thus we have $m = kD_1^2$ and $n = 2(D_1 D + 1)$ and the result follows by \cite[Theorem 5.7]{Porkolab1997}.
\end{proof}
In our application of the algorithm, we assumed that $k \leq D^2$. If we have $k = \OO(1)$ and $D_1 = \OO(1)$, which means that we are interested in just a few effect operators and we have an upper bound on the block dimension uniform in $D$, the SDP can be solved in a number of operations polynomial in $D$. If this is not the case, the performance of the algorithm can be significantly worse. The reason for this is that the separation between the two cases interpolation possible/impossible can become double exponentially small if we bound the operator norm of the $H_i$ we allow. 

\bibliographystyle{alpha}
\bibliography{lit}

\end{document}